\documentclass[11pt]{article}

\usepackage[preprint]{acl}

\usepackage{times}
\usepackage{latexsym}

\usepackage[T1]{fontenc}

\usepackage[utf8]{inputenc}

\usepackage{microtype}

\usepackage{inconsolata}

\usepackage{graphicx}

\usepackage{hyperref}
\usepackage{url}
\usepackage{booktabs}
\usepackage{amsfonts}
\usepackage{amsmath}
\usepackage{amssymb}
\usepackage{amsthm}
\usepackage{nicefrac}
\usepackage{microtype}
\usepackage{xcolor}
\usepackage{algorithm}
\usepackage{algorithmic}
\usepackage{multirow}
\usepackage{makecell}
\usepackage{booktabs}
\usepackage{subcaption}
\usepackage{bbm}
\usepackage{booktabs}
\usepackage{multirow}
\usepackage[normalem]{ulem}

\usepackage{pifont}
\newtheorem{theorem}{Theorem}
\newtheorem{lemma}[theorem]{Lemma}

\newtheorem{assumption}{Assumption}
\newtheorem{remark}{Remark}

\newcommand{\E}{\mathbb{E}}
\newcommand{\KL}{\text{KL}}
\newcommand{\TV}{\text{TV}}
\newcommand{\piref}{\pi_{\text{ref}}}
\newcommand{\pitheta}{\pi_\theta}

\newcommand{\Acal}{\mathcal{A}}
\newcommand{\Lcal}{\mathcal{L}}

\newcommand{\Rphon}{R_{\mathrm{inter}}}
\newcommand{\ganchor}{g_{\mathrm{anchor}}}

\def\framework{AP-GRPO}
\title{\framework: Anchor-Gated Phonetic Alignment \\with Policy Optimization for Pathological Speech Reconstruction}

\author{
 \textbf{Pengfei Zhang\textsuperscript{1}},
 \textbf{Hoang H Nguyen\textsuperscript{2}},
 \textbf{Yutong Song\textsuperscript{1}},
 \textbf{Wenjun Huang\textsuperscript{1}},
\\
 \textbf{Tahmid Imtiaz Imu\textsuperscript{3}},
 \textbf{Henry Peng Zou\textsuperscript{2}},
 \textbf{Jiang Wu},
 \textbf{Honghui Xu \textsuperscript{3}},
 \textbf{Amir M. Rahmani \textsuperscript{1}}
\\
\\
 \textsuperscript{1}University of California Irvine,
 \textsuperscript{2}University of Illinois Chicago,
 \textsuperscript{3}Kennesaw State University
}

\begin{document}
\maketitle 

\begin{abstract}
Pathological speech from patients with neurodegenerative and neuromotor disorders is often acoustically distorted and linguistically fragmented, making pathological speech reconstruction necessary to recover intended textual content from distorted and incomplete speech recordings.
Crucially, such recordings are rarely uniformly degraded: some words or short phrases remain reliable and can serve as \emph{audible anchors} for reconstructing the corrupted surrounding content.
We introduce \textbf{Anchor-gated Phonetic Group Relative Policy Optimization (AP-GRPO)}, a GRPO framework with phonetic reward that aligns speech language models (SLMs) through audible-anchor preservation and inter-anchor phonetic compatibility to the original speech signal.
AP-GRPO consists of: (i) an \textbf{anchor-gated reward} that matches reliable audible anchors in clear regions; and (ii) an \textbf{inter-anchor phonetic alignment reward} that evaluates whether recovered contents are phonetically supported by the corresponding corrupted inter-anchor speech span.
Across four disease conditions, AP-GRPO improves faithful speech reconstruction, and the learned anchor constraint automatically adapts to each condition and thus reveals interpretable disease-specific profiles: conditions with severe articulatory degradation require stronger anchor enforcement, whereas milder impairment or linguistically impaired conditions rely more on phonetic alignment for inter-anchor recovery.

\end{abstract}

\begin{figure}[t]
    \centering
    \includegraphics[width=\linewidth]{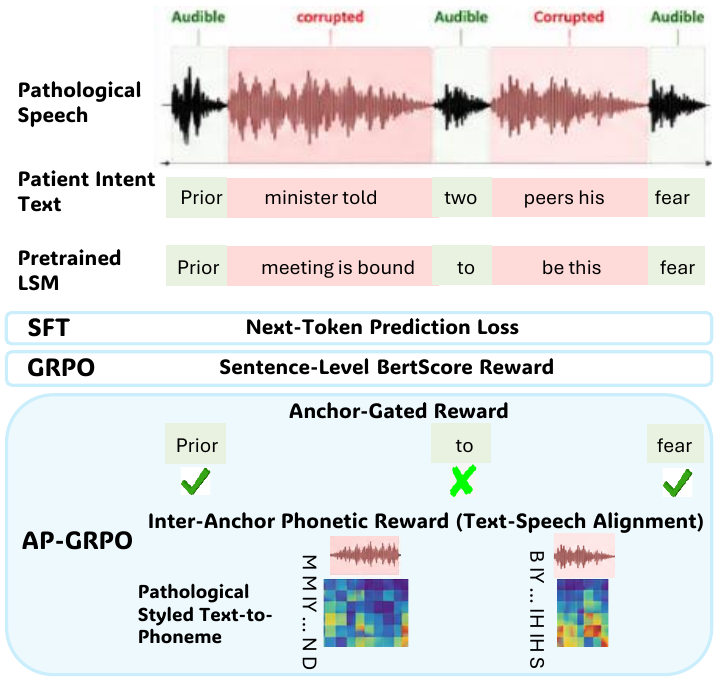}
    \vspace{-0.7cm}
    \caption{\textbf{Conceptual overview of AP-GRPO.} Pathological speech contains both audible spans and corrupted spans. SFT trains via next-token prediction on paired data; vanilla GRPO uses sentence-level BERTScore, which cannot distinguish faithful reconstruction from fluent paraphrasing. AP-GRPO extracts audible anchors as structural boundary conditions and evaluates inter-anchor content by aligning pathologically styled phoneme path against the original speech's signal, targeting precisely the corrupted regions using the patient's own speech as the reference signal.
}
\vspace{-1.5em}
    \label{fig:intro}
\end{figure}

\section{Introduction}

Neurodegenerative and neuromotor disorders such as Parkinson's disease \cite{pakinsondata}, amyotrophic lateral sclerosis (ALS) \cite{torgodata}, cerebral palsy \cite{kim2008uaspeech}, and dementia \cite{luz2021adresso}, frequently impair speech production, limiting a patient's ability to communicate. Unlike conventional Automatic Speech Recognition, \textbf{pathological speech reconstruction (PSR)} must recover a patient's intended text from degraded and sometimes linguistically fragmented recordings \cite{halpern2025dysarthric_overview,qian2023dysarthric_survey}. This distinction is crucial: in severe cases, large acoustic regions may be distorted beyond recognition while the patient's communicative intent remains partially recoverable \cite{song2026demma}.
A key observation, illustrated in Figure~\ref{fig:intro}, is that pathological speech degradation is rarely uniform \cite{duffy2019motor}. Even in a severely distorted recording, some words or short phrases survive intact as \textbf{audible anchors}, while the corrupted inter-anchor spans concentrate the reconstruction errors. Existing approaches do not exploit it: supervised fine-tuning (SFT) treats all positions equally through next-token prediction loss, and sentence-level GRPO \cite{shao2024deepseekmath,guo2025deepseekr1} cannot distinguish faithful reconstruction from fluent paraphrasing - rewarding "ache" as a substitute for "pain" even when the patient clearly produced the latter. 
Other generative reconstruction systems require substantial paired data that is intractable to collect from patients with progressive conditions.

We propose \textbf{Anchor-Gated Phonetic GRPO (AP-GRPO)}, a Group Relative Policy Optimization framework for SLM-based PSR, which addresses these limitations by computing the RL reward based on the alignment to the patient's own speech signal (Figure~\ref{fig:intro}). AP-GRPO first extracts audible anchors as structural boundary conditions and then evaluates each candidate inter-anchor text through phonetic alignment \cite{cuturi2017softdtw}: matching a pathologically styled phoneme sequence from the text against the corresponding speech span. This reward targets precisely the corrupted regions where errors concentrate, using acoustic evidence that is available for every patient recording without any reference transcript.

We evaluate AP-GRPO on four pathological conditions (ALS, cerebral palsy, dementia, Parkinson's). On the most severe conditions, AP-GRPO reduces WER from $\sim 0.75$ to $\sim 0.29$, transforming barely intelligible output into partially correct output suitable for clinical communication and downstream error correction \citep{la2025exploring}. Further analysis shows that AP-GRPO largely suppresses hallucinated insertions which benefits from pathological styling and phonetic alignment. AP-GRPO's optimization automatically adapt to each condition's difficulty, matching clinical impairment profiles without per-disease tuning.




We make three contributions. \textbf{(1)} We introduce AP-GRPO, the first GRPO framework that aligns LSMs to pathological speech. \textbf{(2)} AP-GRPO combines an anchor-gated reward with phonetic compatibility score of pathological-style inter-anchor texts. Experiments demonstrates the effectiveness across all four pathology conditions, with an adaptive mechanism that automatically allocates optimization between anchor preservation and phonetic infilling, which serves as an interpretable proxy for disease-specific reconstruction difficulty.
\textbf{(3)} Stratified Analysis based on disease secerity shows that AP-GRPO improves largely on the severe conditions and can largely suppresses hallucination, opening PSR to clinical settings.

\section{Related Work}

\textbf{Pathological and dysarthric speech reconstruction.}
Unit-DSR \cite{wang2024unitdsr} reconstructs pathological speech through HuBERT-unit \cite{hsu2021hubert} normalization, while CoLM-DSR \cite{chen2024colmdsr} applies codec language modeling to dysarthric speech reconstruction and Diff-DSR \cite{chen2025diffdsr} utilizes Diffusion Models \cite{rombach2022latentdiffusion}. These methods demonstrate the value of generative modeling for pathological speech, but typically rely on substantial paired supervision. 
AP-GRPO uses phonetic evidence from the original speech signal to guide text reconstruction during RL alignment.

\noindent \textbf{RL alignment for speech.}
SpeechAlign \cite{zeng2024speechalign} introduced DPO-style alignment for text-to-speech, and Align-SLM \cite{deng2025alignslm} applied reinforcement learning from AI feedback to spoken language models. These methods primarily rely on sentence-level or preference-level reward signals. AP-GRPO focuses more on phonetic alignment which is important to PSR.

\noindent \textbf{PPGs and phoneme alignment.}
Phoneme posteriorgrams have been used for voice conversion \cite{sun2016ppgvc} and pronunciation assessment \cite{liu2024ppgassess}. Soft-DTW \cite{cuturi2017softdtw} provides a differentiable sequence-alignment objective. We combine PPGs with Soft-DTW to evaluate whether generated inter-anchor text is phonetically supported by the speech.


\begin{figure*}[t]
    \centering
    \includegraphics[width=\linewidth]{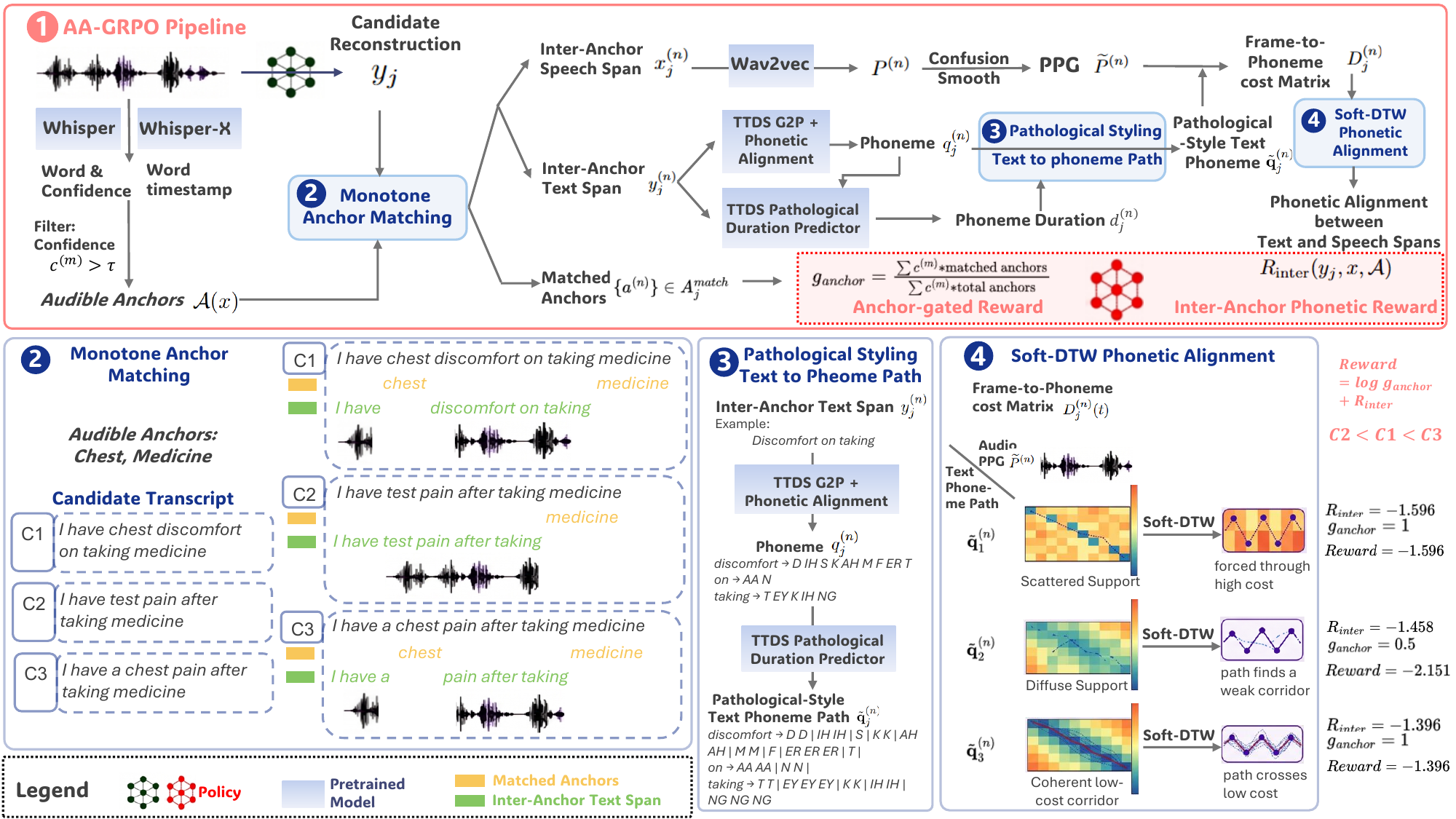}
    \vspace{-0.7cm}
    \caption{\textbf{Pipeline of AP-GRPO}. Audible Anchors are extracted from speech using Whisper/WhisperX with confidence $>\tau$. At RL steps, the policy samples a group of candidate reconstructions (Step \textcircled{2}). For each candidate, AP-GRPO performs anchor matching and thus calculate the inter-anchor text and speech spans. Anchor-gated reward is computed from anchor coverage. In parallel, each inter-anchor text is converted into a pathological-style phoneme path using Pretrained TTDS (Step \textcircled{3}), and corresponding phoneme posteriorgram (PPG) are obtained for each inter-anchor speech span. Phonetic alignment score is measured between the PPG and the candidate phoneme using Soft-DTW (Step \textcircled{4}). AP-GRPO combines both rewards. The example shows three sampled candidates: C1 preserves the anchors but largely acoustically mismatched, C2 weakly matches the left anchor, and C3 best preserves the anchors and yields the lowest cost phonetic alignment, resulting in the highest reward. 
}
    \label{fig:overview}
\end{figure*}

\section{AP-GRPO}
\label{sec:method}

AP-GRPO aligns a speech large language model with pathological speech recordings by using reliable audible fragments as anchors and by scoring the corrupted text between anchors against the original speech signal. Given a pathological speech recording $x$, the policy $\pitheta(\cdot|x)$ generates candidate reconstruction texts. AP-GRPO evaluates each candidate with two components: an anchor-gated reward for preserving reliable audible evidence, and an inter-anchor phonetic reward for checking whether recovered text is supported by the corresponding speech span.


\subsection{Audible Anchors and Candidate Spans}

At Data Preprocessing Stage, for each pathological speech recording $x$, we use whisper-X (\citet{bain2023whisperx, radford2023whisper}) to extract an ordered set of timestamped audible anchors:
\[
\Acal(x)=\{(a^{(m)},s^{(m)},e^{(m)},c^{(m)})\}_{m=1}^{M}, c^{(m)}>\tau.
\]
Here, $a^{(m)}$ is a reliable word or short phrase, $s^{(m)}$ and $e^{(m)}$ are its start and end times, and $c^{(m)}$ is its confidence. Anchors are the word segments with confidence larger than a threshold.

At the RL Stage, the LSM $\pi(\cdot|x)$ samples a group of candidate reconstruction texts,
\[
\{y_1,\ldots,y_G\}\sim\pi(\cdot|x).
\]
Each $y_j$ is a complete hypothesis for the intended text. We align anchors to $y_j$ in their temporal order using monotonic anchor matching: anchors and candidate words are scanned left to right, and each anchor is matched to the earliest compatible unmatched candidate span after the previous match.

For candidate text $y_j$, the region between each pair of matched boundary anchors are defined as an inter-anchor span. We denote the set of such regions by $\mathcal{I}_j$. For each $n\in\mathcal{I}_j$, let $y_j^{(n)}$ be the candidate text span between the matched boundary anchors and $x^{(n)}$ be the corresponding speech span. Intuitively, $y_j^{(n)}$ is the model's recovered text for a region where the original recording is less reliable. Details of matching are provided in Appendix~\ref{app:anchor_extraction}.

\subsection{Pathology-Aware Inter-Anchor Phonetic Alignment}

The most tricky part of PSR is to recover the corrupted text between anchors. A generic text-level reward is poorly suited to this step: it may reward fluent paraphrases that are semantically plausible but not supported by what the patient actually produced. AP-GRPO therefore uses a phonetic alignment reward that evaluates 
how a candidate inter-anchor text $y_j^{(n)}$ can be phonetically supported by the speech span $x^{(n)}$ in pathological scenarios.

\paragraph{Phoneme Evidence from Pathological Speech.}
A frozen wav2vec2 phoneme model converts the speech span $x^{(n)}$ into a phoneme posteriorgram
\[
P^{(n)}\in\mathbb{R}^{T_n\times K},
\]
where $T_n$ is the number of acoustic frames and $K$ is the phoneme inventory size. Each entry $P^{(n)}(t,k)$ estimates the support for phoneme $k$ at frame $t$.

Pathological speech often contains systematic pronunciation deviations: a clinically faithful production may differ from the canonical phoneme realization. For instance, dysarthric articulation may blur voicing contrasts such as /t/ versus /d/, shift fricatives such as /s/ toward /sh/, or centralize vowels such as /i/ toward /ih/ or /ax/ (see Appendix~\ref{app:confusion} full list). To avoid over-penalizing such cases, we smooth the posteriorgram with a row-normalized phoneme confusion matrix $C$:
\begin{equation}
    \widetilde P^{(n)}(t,q)
    =
    \sum_{k=1}^{K} C(q,k)P^{(n)}(t,k).
    \label{eq:confusion_smooth}
\end{equation}
Here, $q$ denotes the candidate phoneme from the text-side phoneme path. The diagonal value $C(q,q)$ is assigned the largest weight, so exact phoneme matches receive the strongest support. Phonetically similar or clinically confusable phonemes are assigned smaller but nonzero weights.
Thus, smoothed PPG $\widetilde P^{(n)}(t,q)$ measures dysarthria-tolerant support for candidate phoneme $q$ at frame $t$, rather than strict canonical phoneme probability.

\paragraph{Pathological-Styling Text-to-Phoneme Path.}
Another important characteristic of pathological speech is that they may contain slowed articulation, prolonged phones, irregular pauses, and unstable rhythm. We therefore convert the candidate text span $y_j^{(n)}$ into a phoneme sequence $q_{j}^{(n)}=[q_{j,1}^{(n)},\ldots,q_{j,N_n}^{(n)}]$, and expand it using a dysarthric duration prior modeled by a pretrained TTDS model (\cite{leung2024ttds}):
\[
\tilde{\mathbf q}_{j}^{(n)}
=
[
\underbrace{q_{j,1}^{(n)},\ldots,q_{j,1}^{(n)}}_{d_{j,1}^{(n)}},
\ldots,
\underbrace{q_{j,N_n}^{(n)},\ldots,q_{j,N_n}^{(n)}}_{d_{j,N_n}^{(n)}}
].
\]
The duration vector $d_{j,N_n}^{(n)},\ldots,d_{j,N_n}^{(n)}$ is predicted by a TTDS-style dysarthric duration model. This expansion does not force an exact timing match; instead, it places the candidate phoneme sequence at an approximate frame-level resolution, while Soft-DTW still allows flexible monotonic warping.

\paragraph{Soft-DTW Phonetic Alignment.}
We build a frame-to-path cost matrix between each phoneme in the candidate pathological-styling text-to-phoneme path and smoothed PPG: 
\begin{equation}
    D_j^{(n)}(t,\ell)
    =
    -
    \log
    \left(
    \widetilde P^{(n)}(t,\tilde q_{j,\ell}^{(n)})
    +
    \epsilon
    \right).
    \label{eq:cost}
\end{equation}
Soft-DTW algorithm, proposed by \citet{cuturi2017softdtw}, computes a differentiable monotonic alignment over this matrix. Specifically, a candidate's span receives a high score when its phonemes can be aligned to frames where the pathological speech signal provides strong support:
\begin{equation}
    R_{\mathrm{span},j}^{(n)}
    =
    -
    \frac{
    \mathrm{SoftDTW}_\gamma(D_j^{(n)})
    }{
    T_n+L_j^{(n)}
    }.
    \label{eq:span_score}
\end{equation}
The inter-anchor phonetic reward averages over all valid spans:
\begin{equation}
    \Rphon(y_j,x,\Acal)
    =
    \frac{1}{|\mathcal{I}_j|}
    \sum_{n\in\mathcal{I}_j}
    R_{\mathrm{span},j}^{(n)}.
    \label{eq:r_inter}
\end{equation}
This reward is pathology-aware. It focuses on inter-anchor regions where reconstruction errors concentrate, tolerates clinically plausible phoneme confusions, and accounts for pathological duration patterns. As a result, AP-GRPO rewards candidate texts that are not merely fluent, but phonetically grounded in the patient's original speech recording. Principles of Soft-DTW from \cite{cuturi2017softdtw} are presented in Appendix~\ref{app:ppg_softdtw_background}.

\subsection{Anchor-Gated Reward}

To ensure trusted recovery of audible anchors, an anchor-gated reward $\ganchor(y_j,\Acal)$ is also applied, which is measured as a confidence-weighted coverage score on how many anchors are preserved:
\begin{equation}
     Match_{\Acal}(y_j)
    =
    \frac{
    \sum_{m=1}^{M} c^{(m)} \mathbbm{1}[a^{(m)}\text{ matched }y_j]
    }{
    \sum_{m=1}^{M} c^{(m)}
    },
    \label{eq:gate1}
\end{equation}
\vspace{-6pt}
\begin{equation}
    \ganchor(y_j,\Acal)
    =
    \epsilon_g
    +
    (1-\epsilon_g)
    Match_{\Acal}(y_j),
    \label{eq:gate2}
\end{equation}
where $\epsilon_g>0$ is a small constant that keeps the gate bounded away from zero.
By assigning confidence as weights to recovered anchors, $\ganchor(y_j,\Acal)$ discourages fluent but unfaithful candidate texts that contradicts reliable audible anchors.

\subsection{Reward and Alignment Objective}

AP-GRPO follows the standard GRPO training recipe: for each speech recording $x$, the current policy samples a group of candidate texts, rewards are computed for each candidate, and the policy is updated using within-group relative advantages. 
Instead of using a generic text-based reward, AP-GRPO uses an anchor-gated phonetic reward grounded in the original pathological speech signal:
\begin{equation}
    R_{j}(y_j,x,\Acal)
    =
    \log \ganchor(y_j,\Acal)
    +
    \Rphon(y_j,x,\Acal).
    \label{eq:reward}
\end{equation}
with anchor coverage implicitly constraint for the alignment:
\begin{equation}
    \E_{\pitheta}[\ganchor]\geq\alpha.
    \label{eq:anchor_constraint}
\end{equation}
We set $\alpha=0.95$ to enforce that most candidate texts should cover all anchors, as the anchors are the most reliable information within the pathological speeches.
During GRPO training stage, this constraint is enforced using an adaptive dual variable $\mu$. 
The per-candidate reward $R_j=R_{j}(y_j,x,\Acal)$ and the constraint signal $g_j$ have different variance profiles. We therefore normalize them independently within the sampled group, following GRPO:
\[
\hat R_{\mathrm{norm},j}
=
\frac{R_j-\bar R}{\hat\sigma_R+\epsilon},
\qquad
\hat g_{\mathrm{norm},j}
=
\frac{g_j-\bar g}{\hat\sigma_g+\epsilon}.
\]
The advantage for candidate text $y_j$ is
\begin{equation}
    A_j
    =
    \hat R_{\mathrm{norm},j}
    +
    \mu \hat g_{\mathrm{norm},j}.
    \label{eq:advantage}
\end{equation}
The advantage $A_j$ tells us how much better candidate $y_j$ is on AP-GRPO's Reward, 
relative to the group. The policy is then updated by maximizing 
the GRPO objective.
All theoretical proof on the convergence and GRPO details,
are provided in Appendix~\ref{sec:optimization} and \ref{app:proofs}.

\begin{table*}[t]
\centering
\footnotesize
\caption{Main results of PSR on four datasets. \textbf{Bold numbers} indicate the best performance, and \underline{underlined numbers} indicate the second-best performance.}
\label{tab:four_dataset_results}
\vspace{-0.3cm}
\resizebox{\textwidth}{!}{
\setlength{\tabcolsep}{5pt}
\begin{tabular}{@{}c@{\hspace{0.5em}}c@{}}
\toprule

\begin{tabular}{lcccc}
\multicolumn{5}{c}{\textit{Amyotrophic Lateral Sclerosis (ALS) - Dataset: TORGO}} \\
\cmidrule(lr){1-5}
Method & WER $\downarrow$ & CER $\downarrow$ & BLEU-4 $\uparrow$ & CONTENT F1 $\uparrow$ \\
\midrule
Align-SLM & 0.8101 & 0.7894 & 0.3133 & 0.4967 \\
Diff-DSR & 1.0191 & 0.7442 & 0.1138 & 0.3075 \\
Colm-DSR & 0.9533 & 0.7698 & 0.1087 & 0.3168 \\
Gemini-Flash-3.5 & 0.4919 & 0.3022 & 0.2324 & 0.4558  \\
\midrule
Audio-Flamingo 3 & 0.7261 & 0.4413 & 0.3305 & 0.5569  \\
Audio-Flamingo 3 + SFT & 0.5699 & 0.5461 & 0.3577 & 0.5709 \\
Audio-Flamingo 3 + AP-GRPO & \underline{0.3577} & \underline{0.2345} & \underline{0.4409} & \underline{0.6364}  \\
\midrule
Qwen2.5-Omni & 1.2570 & 1.0884 & 0.2255 & 0.4175  \\
Qwen2.5-Omni + SFT & 0.7495 & 0.7566 & 0.3599 & 0.5725 \\
Qwen2.5-Omni + AP-GRPO & \textbf{0.2885} & \textbf{0.1764} & \textbf{0.4573} & \textbf{0.6597} \\
\midrule
\multicolumn{5}{c}{\textit{Cerebral Palsy (CP) - Dataset: UASpeech}} \\
\cmidrule(lr){1-5}
Method & WER $\downarrow$ & CER $\downarrow$ & BLEU-4 $\uparrow$ & CONTENT F1 $\uparrow$ \\
\midrule
Align-SLM & 0.6299 & 0.5574 & 0.3155 & 0.4231  \\
Diff-DSR & 0.7044 & 0.6213 & 0.2282 & 0.3545  \\
Colm-DSR & 0.6461 & 0.6525 &  0.2015 & 0.3767  \\
Gemini-Flash-3.5 & 0.5355 & 0.3617 & 0.3596 & 0.5030  \\
\midrule
Audio-Flamingo 3 & 2.4832 & 1.8437 & 0.4214 & 0.5218 \\
Audio-Flamingo 3 + SFT & 0.3755 & \underline{0.2960} & 0.5286 & 0.6205  \\
Audio-Flamingo 3 + AP-GRPO & \textbf{0.2862} & \textbf{0.2086} & \textbf{0.6418} & \textbf{0.7254} \\      
\midrule
Qwen2.5-Omni & 2.7681 & 2.5978 & 0.1890 & 0.4541 \\
Qwen2.5-Omni + SFT & 0.4511 & 0.4660 & 0.6224 & 0.7214 \\
Qwen2.5-Omni + AP-GRPO & \underline{0.3696} & 0.3462 & \underline{0.6249} & \underline{0.7225}  \\
\end{tabular}

&

\setlength{\tabcolsep}{5pt}
\begin{tabular}{lcccc}
\multicolumn{5}{c}{\textit{Dementia - Dataset: ADReSSo}} \\
\cmidrule(lr){1-5}
Method & WER $\downarrow$ & CER $\downarrow$ & BLEU-4 $\uparrow$ & CONTENT F1 $\uparrow$ \\
\midrule
Align-SLM & 0.2968 & 0.1801 & 0.4485 & 0.6530  \\
Diff-DSR & 0.3156 & 0.2188 & 0.4045 & 0.6090  \\
Colm-DSR & 0.2900 & 0.1770 & 0.4572 & 0.6617  \\
Gemini-Flash-3.5 & 0.2586 & 0.1379 & 0.5286 & 0.8824  \\
\midrule
Audio-Flamingo 3 & 0.3541 & 0.1484 & 0.5000 & 0.8440 \\
Audio-Flamingo 3 + SFT & 0.2909 & 0.1336 & 0.5719 & 0.8820 \\
Audio-Flamingo 3 + AP-GRPO & \textbf{0.2075} & \textbf{0.1141} & \textbf{0.6589} & \textbf{0.9083} \\
\midrule
Qwen2.5-Omni & 0.3858 & 0.1599 & 0.4789 & 0.8190 \\
Qwen2.5-Omni + SFT & 0.3090 & 0.1350 & 0.5519 & 0.8760 \\
Qwen2.5-Omni + AP-GRPO & \underline{0.2234} & \underline{0.1199} & \underline{0.6472} & \underline{0.8983}\\
\midrule
\multicolumn{5}{c}{\textit{Parkinson - Dataset: SJTU Parkinson Patient Speech Dataset
}} \\
\cmidrule(lr){1-5}
Method & WER $\downarrow$ & CER $\downarrow$ & BLEU-4 $\uparrow$ & CONTENT F1 $\uparrow$ \\
\midrule
Align-SLM & 0.2284 & 0.0957 & 0.6157 & 0.9055 \\
Diff-DSR & 0.2312 & 0.1161 & 0.5986 & 0.8835 \\
Colm-DSR & 0.2160 & 0.0930 & 0.5950 & 0.8864 \\
Gemini-Flash-3.5 & 0.1817 & 0.0925 & 0.5878 & 0.9023 \\
\midrule
Audio-Flamingo 3 & 0.2426 & 0.0949 & 0.6109 & 0.9022 \\
Audio-Flamingo 3 + SFT & \underline{0.1615} & \underline{0.0717} & \underline{0.6961} & 0.9079 \\
Audio-Flamingo 3 + AP-GRPO & 0.1682 & 0.0783 & 0.6938 & \underline{0.9107} \\
\midrule
Qwen2.5-Omni & 0.2090 & 0.1058 & 0.6400 & 0.8847 \\
Qwen2.5-Omni + SFT & 0.1724 & 0.0799 & 0.6856 & 0.9087 \\
Qwen2.5-Omni + AP-GRPO & \textbf{0.1592} & \textbf{0.0701} & \textbf{0.6992} & \textbf{0.9141} \\

\end{tabular}

\\
\bottomrule
\end{tabular}
}

\end{table*}

\begin{figure*}[t]
\vspace{-0.3cm}
    \centering
    \includegraphics[width=\linewidth]{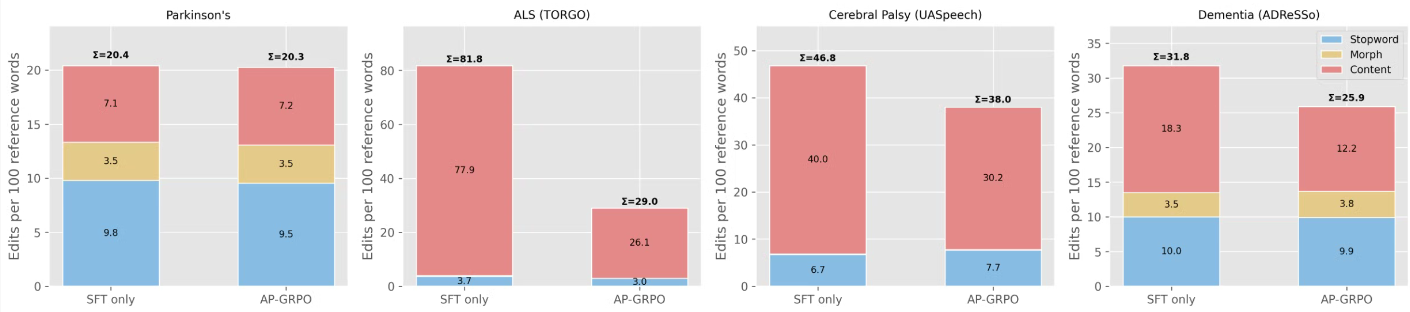}
    \vspace{-0.9cm}
    \caption{WER edits per 100 reference words for SFT-only vs AP-GRPO hypotheses, decomposed into Stopword, Morph, and Content edits. Stopword and Morph errors are invisible to Content-F1 - only the Content errors hurt both metrics.
    All tests are measured from the best AP-GRPO checkpoint on Qwen2.5-Omni-7B.}
    \label{fig:wer_cf1_asym}
\end{figure*}

\section{Experiments}
\label{sec:experiments}

\noindent \textbf{Datasets.} We evaluate across four datasets of different neurodegenerative conditions: ADReSSo \cite{luz2021adresso} on Dementia/Alzheimer's Disease, TORGO \cite{torgodata,rudzicz2012torgo2} on Amyotrophic Lateral Sclerosis (ALS), UASpeech \cite{kim2008uaspeech} on Cerebral Palsy, and SJTU Parkinson Speech Dataset \cite{pakinsondata} on Pakinson. Full introductions are presented at Appendix~\ref{app:data_stats}.


\noindent \textbf{Metrics.}
We evaluate the reconstructed text on both accuracy and semantic fidelity: \textbf{WER} (word error rate, $\downarrow$) for transcript quality; \textbf{CER} (character error rate, $\downarrow$) for character-level accuracy; \textbf{BLEU-4} ($\uparrow$), 4-gram precision with brevity penalty for measuring n-gram overlap; and \textbf{Content-F1} ($\uparrow$), token-level F1 on content words (nouns, verbs, adjectives, adverbs), excluding function words, measuring whether substantive content is preserved.

\noindent \textbf{Baseline.} We compare to: 1) previous work on PSR, including 
diffusion-based Diff-DSR \cite{chen2025diffdsr} and Codec LM-based Colm-DSR \cite{chen2024colmdsr}; 2) large speech models, including Qwen2.5-Omni \cite{xu2025qwen25omnitechnicalreport}, Audio-Flamingo 3 \cite{ghosh2026audio} and Gemini-Flash-3.5 \cite{comanici2025gemini}; 3) other Reinforcement Learning methods including Align-SLM \cite{deng2025alignslm}. As AP-GRPO is applied to Qwen2.5-Omni and Audio-Flamingo 3, we also compare against the SFT-only method on these models.



\noindent \textbf{Implementation Details.} 
For each dataset, we split speakers into train/validation/test partitions with a 70/10/20 ratio, using the held-out test split used only for final evaluation. The validation split is used to monitor alignment dynamics
and to select the checkpoint. Transcript of Speech are used on SFT stage. No ground truth transcript is used on GRPO stage. Anchors are extracted once before RL training using Whisper large-v3 with confidence threshold $\tau=0.85$, followed by WhisperX timestamp refinement; we retain anchors with at least three characters and cache the resulting timestamped anchor set for all train, validation, and test recordings. AP-GRPO is initialized from the SFT policy and trained for one alignment epoch of GRPO steps equal to $\frac{dataset~size}{mini~batch~size}$ with group size $G=32$.
For full details, see Appendix~\ref{app:implementation}.

\section{Results}

\subsection{Main Results}

Table 1 summarizes the performance of AP-GRPO compared to previous methods. AP-GRPO achieves the lowest WER and CER on all four datasets for at least one backbone, with the largest gains on the most challenging conditions:
 
\paragraph{Severe dysarthria: from barely understandable to partially correct.}
On TORGO (ALS), AP-GRPO reduces WER from the SFT baseline of ${\sim}0.75$ and ${\sim}0.57$ down to $0.36$ and $0.29$, respectively - the largest absolute improvements in the table. At WER $> 0.5$, the majority of words in a reconstruction are wrong; the output is unusable for clinical or communicative purposes. At WER ${\sim}0.3$, roughly two-thirds of words are correct: the reconstruction preserves sentence structure and most content words, making it feasible both for healthcare communication (patients and caregivers can verify intent) and as input to downstream LLM-based error correction \citep{la2025exploring}, where a partially correct transcript is far more recoverable than a barely intelligible one. 

\begin{table*}[t]
\centering
\caption{Reward ablation on Qwen 2.5-Omni on Torgo. 
"$\log \ganchor$" = anchor-gated reward. 
"$\Rphon$" = inter-anchor phonetic reward. 
"Cons" = confusion smoothing for calculating $\Rphon$. 
"PaS" = pathological styling based on pretrained TTDS. 
"$\mu$" = adaptive dual variable for anchor constraint $\E_{\pitheta}[\ganchor]\geq\alpha$.}
\label{tab:reward_ablation}
\small
\vspace{-0.3cm}
\resizebox{\textwidth}{!}{
\setlength{\tabcolsep}{5pt}
\begin{tabular}{@{}lccccccccc@{}}
\toprule
Config & $\log \ganchor$ & $\Rphon$ & Cons & PaS & $\mu$ 
  & WER $\downarrow$ & CER $\downarrow$ & BLEU-4 $\uparrow$ & CONTENT F1 $\uparrow$ \\
\midrule
SFT-Only
  & --- & --- & --- & --- & --- & 0.7495 & 0.7566 & 0.3599 & 0.5725 \\
GRPO (Reward = BERTScore) 
  & \texttimes & \texttimes & --- & --- & --- & 0.7457 & 0.7540 & 0.3691 & 0.5178 \\
\midrule
\multicolumn{9}{l}{\textit{AP-GRPO Reward Settings}} \\
Anchor-Gate only ($\log g$, no $\Rphon$)
  & \checkmark & \texttimes & --- & --- & adaptive & 0.7466 & 0.7527 &   0.3637 & 0.5742 \\
Phonetic Reward only ($\Rphon$, no $\log g$)
  & \texttimes & \checkmark & \checkmark & \checkmark & adaptive &  0.6901 & 0.6425 & 0.2012 & 0.2811 \\
Anchor-Gate + smoothed PPG + PaS
  & \checkmark & \checkmark & \checkmark & \checkmark & adaptive & \textbf{0.2885} & \textbf{0.1764} & \textbf{0.4573} & \textbf{0.6597} \\
Anchor-Gate + smoothed PPG (no PaS)
  & \checkmark & \checkmark & \checkmark & \texttimes & adaptive & 0.7404 & 0.7494 & 0.3682 & 0.5814 \\
Anchor-Gate + raw PPG + PaS 
  & \checkmark & \checkmark & \texttimes & \checkmark & adaptive & 0.3725 & 0.2461 & 0.4086 & 0.6235 \\
\midrule
\multicolumn{9}{l}{\textit{Constraint mechanism}} \\
No constraint ($\mu = 0$ fixed) 
  & \checkmark & \checkmark & \checkmark & \checkmark & none & 0.4932 & 0.4435 & 0.3997 & 0.5793 \\
Fixed $\mu = 1.0$ 
  & \checkmark & \checkmark & \checkmark & \checkmark & fixed & 0.4699 & 0.3077 & 0.4259 & 0.5955 \\
\textbf{Adaptive $\mu$ (momentum)} 
  & \checkmark & \checkmark & \checkmark & \checkmark & adaptive 
  & \textbf{0.2885} & \textbf{0.1764} & \textbf{0.4573} & \textbf{0.6597} \\
\bottomrule
\end{tabular}
}
\end{table*}

\vspace{-0.2cm}
\begin{figure*}[t]
\vspace{-0.3cm}
\centering

\begin{subfigure}{0.15\textwidth}
    \centering
    \includegraphics[width=\linewidth]{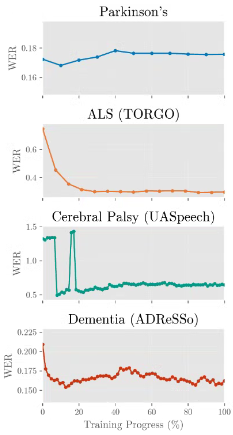}
    \vspace{-0.6cm}
    \caption{Learning Curve}
    \label{fig:left}
\end{subfigure}
\hfill
\begin{subfigure}{0.4\textwidth}
    \centering
    \includegraphics[width=\linewidth]{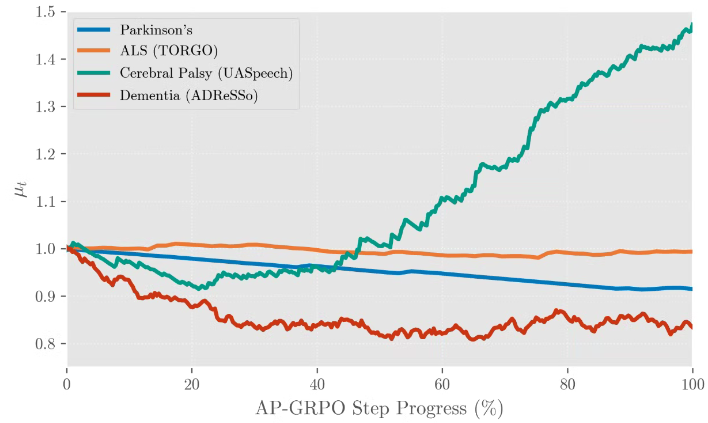}
    \vspace{-0.6cm}
    \caption{Learned dual variable $\mu$ for anchor coverage constraint over the time}
    \label{fig:right}
\end{subfigure}
\hfill
\begin{subfigure}{0.4\textwidth}
    \centering
    \includegraphics[width=\linewidth]{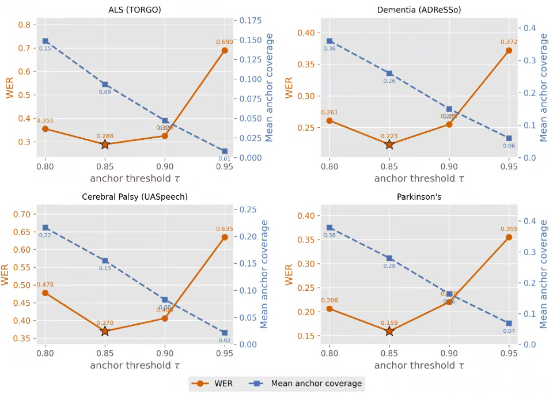}
    \vspace{-0.6cm}
    \caption{WER performance and mean anchor coverage at different anchor threshold}
    \label{fig:right}
\end{subfigure}
\vspace{-0.3cm}
\caption{Different AP-GRPO training configs and trends}
\label{fig:learning_curve}
\end{figure*}

\begin{figure}[t]
    \centering
    \includegraphics[width=\linewidth]{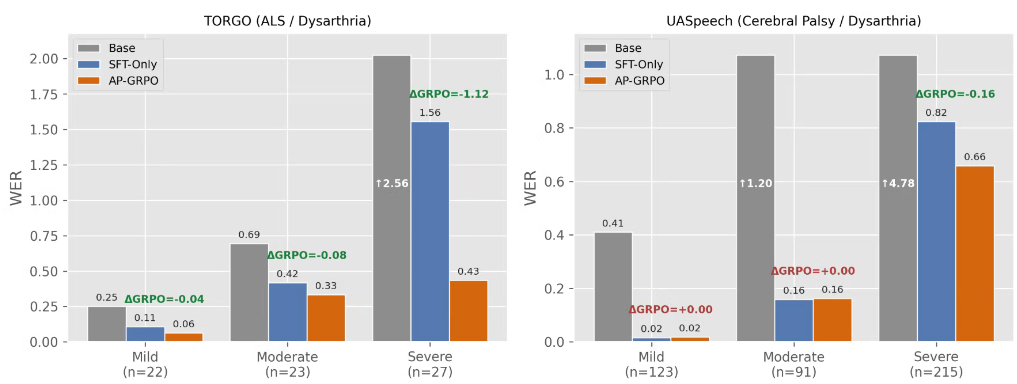}
    \vspace{-0.8cm}
    \caption{Severity-stratified Analysis.
}
    \label{fig:severity}
\end{figure}

\paragraph{Mild dysarthria: consistent but modest gains.}
On UASpeech (Cerebral Palsy), ADReSSo (dementia) and SJTU-PD (Parkinson's), SFT baselines already achieve bearable WER $0.17$--$0.45$. AP-GRPO can still achieve substantial gains on WER. 

 
\paragraph{AP-GRPO benefits both backbones.}
AP-GRPO improves over SFT on both Qwen2.5-Omni and Audio-Flamingo~3 across all four datasets, demonstrating that the method is backbone-agnostic. On TORGO, ADReSSo and Parkinson's dataset, AP-GRPO achieves best performance on Qwen2.5-Omni, and on UASpeech AP-GRPO achieves best performance on Audio Flamingo 3.

\paragraph{AP-GRPO suppresses hallucinated insertions.} On TORGO and UASpeech, AP-GRPO reduces WER by $0.46$ and $0.1$ but Content-F1 improves by only $0.09$ and $0.01$. To understand this asymmetry, we decompose every WER error into three mutually exclusive categories: \textbf{Content} (lexical substitutions, insertions, or deletions of meaningful content words), \textbf{Morph} (same lemma, different surface form, e.g, controls -> controlling, runs -> running), and \textbf{Stopword} (function words, e.g., "a", "the", etc). All three categories contribute to WER edits, but Content-F1 only accounts for Content edits; Morph and Stopword edits are filtered out when calculating Content-F1.
Figure~\ref{fig:wer_cf1_asym} reveals the source of AP-GRPO's improvement: a sharp reduction in content-word errors. On TORGO, content edits drop from $77.9$ to $26.1$ per 100 reference words. The SFT baseline hallucinates extra words on short prompts (e.g., ref = "bat" $\to$ pred = "bat bay a at"), which degrades WER more than Content F1, and AP-GRPO suppresses these fabrications. 
Parkinson's is not affected, confirming that the effect is specific to hallucination-heavy conditions.
This suppression of hallucination is an emerging property of phonetic alignment reward: hallucinated words produce phoneme sequences that are not supported by the speech, receiving low Soft-DTW scores. Importantly, AP-GRPO uses pathological styling which successfully aligns correct texts and suppresses redundant incorrect words.

\paragraph{Convergence.}
Learning curve of AP-GRPO on validation dataset is shown at Figure~\ref{fig:learning_curve} (a). AP-GRPO converges at first 40 steps across all datasets, smaller than 20\% of the total steps.

\subsection{Reward Ablation}

Table~\ref{tab:reward_ablation} isolates the contribution of each design component.
The ablation reveals that AP-GRPO's improvement is not attributable to any single component but to the interaction of each design: the anchor gate, the phonetic reward, pathological styling, and the adaptive learned constraint $\mu$ for applying anchor coverage constraint (Equation~\ref{eq:anchor_constraint}).

\paragraph{Semantic Similarity (BERTScore) reward is not effective for content fidelity inPSR.}
Replacing AP-GRPO's reward with BERTScore (row~2) \cite{zhang2019bertscore} as reward for GRPO slightly improves WER over SFT. BERTScore rewards semantically similar paraphrases (i.e., substituting "discomfort" for "pain"), which inflates the semantic similarity score while degrading exact text recovery. This confirms that text-level semantic rewards are misaligned with the PSR objective.

\paragraph{Two reward components are complementary.}
Using only the anchor-gated reward (row~3) produces WER nearly identical to SFT ($0.747$ vs.\ $0.750$): anchors are the most identifiable signal within the speech, therefore gate values are similar across candidates and the normalized advantage carries almost no learning signal. Conversely, using only the phonetic reward (row~4) reduces WER to $0.690$ - the model generates phoneme sequences that align with the speech PPG but forgets to form coherent words. The two components are complementary: the anchor-gated reward structures the reward around reliable audible evidence, while the phonetic reward provides the discriminative signal for recovering corrupted content between anchors. Removing either one breaks the method.

\paragraph{Pathological speech adaptations are critical for reward quality.}
The two components designed specifically for pathological speech, confusion smoothing and pathological styling, together account for the largest performance gains in the table. Pathological styling (PaS) addresses the temporal distortions common in motor speech disorders. Without it, the candidate phoneme path uses standard durations, creating a severe length mismatch against the audio PPG and making Soft-DTW non-discriminative across candidates (row~7, WER $= 0.740$ vs.\ row~6, WER $= 0.289$). Confusion smoothing addresses systematic dysarthric articulation by assigning nonzero values to phonetically similar or clinically confusable phonemes. Adding smoothing on top of PaS reduces WER from $0.373$ to $0.289$ (rows~8 vs.~6). 

\paragraph{The adaptive constraint closes the remaining gap.}
Rows~10 to 11 use identical rewards but differ in how $\mu$ is handled: $\mu=0$ achieves WER $0.493$, fixed $\mu=1.0$ achieves $0.470$, and adaptive $\mu$ achieves $0.289$. The key advantage of adaptive $\mu$ is that it varies over training: when anchor coverage drops below $\alpha$, $\mu$ increases temporarily to prioritize anchor recovery; when most anchors are recovered, $\mu$ decreases to let the phonetic reward drive inter-anchor improvement. A fixed $\mu$ cannot provide this dynamic balance: if set too low, the model drops anchors; if set too high, the phonetic reward is suppressed and inter-anchor content stagnates. The adaptive mechanism oscillates between these states, allocating capacity to whichever objective needs it at each training step. Figure~\ref{fig:learning_curve} (b) shows the learning curve of $\mu$ over training.

\subsection{Anchor Threshold Sensitivity}
Figure~\ref{fig:learning_curve}(c) reveals WER curves at different confidence threshold, with all optimum values at $\tau = 0.85$.
The two sides of the curve fail for different reasons. At low $\tau$ ($0.80$), more anchors are retained but some are misrecognized or poorly timestamped; these false anchors corrupt both the gate (rewarding incorrect words) and the phonetic reward (aligning candidates against wrong audio segments). At high $\tau$ ($0.90$--$0.95$), anchor coverage drops below 5\%, leaving less anchors that can be referred and inter-anchor spans too long for AP-GRPO to discriminate among candidates, leading to TORGO's WER jumps from $0.288$ to $0.690$. The optimum $\tau = 0.85$ balances anchor reliability against anchor density, and its consistency across all four conditions suggests it reflects Whisper's confidence calibration rather than a disease-specific property, requiring no per-disease tuning.
The WER degradation at strict $\tau$ is sharpest for severe conditions (TORGO: $+0.40$ from $\tau = 0.85$ to $0.95$) and mildest for conditions with abundant anchors (Parkinson's: $+0.20$), consistent with the $\mu_t$ ordering in Figure~\ref{fig:learning_curve}(b): conditions that require stronger anchor enforcement are also most sensitive to anchor threshold selection. Detailed anchor statistics and analysis across four datasets are shown at Figure~\ref{fig:anchor_threshold_stats}.


\subsection{Severity-Stratified Analysis}
We further stratifies results by speaker severity on the two datasets that contain severity annotations, shown at Figure~\ref{fig:severity}.
AP-GRPO's advantage scales with impairment severity. On both datasets, AP-GRPO provides little gain on mild and moderate speakers but largely reduces the WER on severe conditions by $-1.12$ and $-0.16$, respectively, a shift from unintelligible to partially recoverable. 
The pattern is consistent: the phonetic reward contributes most when inter-anchor spans are long and conventional fine-tuning has the least to work with.

Now let's look back to $\mu$ trajectories in Figure~\ref{fig:learning_curve}(b). Cerebral palsy (UASpeech) rises highest to $\mu \approx 1.5$, reflecting the most severe articulatory distortion. $\mu$ for ALS (TORGO) oscillates around $1.0$, also reflecting substantial motor-speech degradation. The ordering, CP $>$ ALS $>$ Parkinson's $>$ dementia, reflects each condition's impact on articulatory precision rather than overall reconstruction difficulty. Notably, dementia produces a lower $\mu$ than Parkinson's despite having higher WER ($0.21$ vs.\ $0.16$): dementia patients typically speak clearly but produce semantically impaired content, making anchors easy to preserve but inter-anchor recovery hard, whereas Parkinson's degrades articulation while preserving communicative intent. This dissociation between $\mu$ and WER demonstrates that $\mu$ captures a specific dimension of pathological speech, audible anchor recognition difficulty, rather than simply tracking overall task difficulty.

\noindent \textbf{Qualitative Examples} are presented at Appendix~\ref{app:qualitative}







\section{Conclusions}

We presented AP-GRPO, which aligns speech language models for pathological speech reconstruction using anchor-gated phonetic rewards derived from the patient's own speech signal. By combining anchor coverage with phonetic alignment scoring and pathological styling, 
AP-GRPO largely improves PER on severe dysarthria, 
crossing the threshold from unintelligible to clinically usable output. The learned anchor constraint automatically discovers disease-specific reconstruction difficulty, aligning with clinical impairment profiles. 
These results suggest that structuring RL rewards around patient-derived acoustic evidence offers a practical path toward accessible communication for people with neurodegenerative speech disorders.
\section{Limitations}
\label{sec:limitations}

AP-GRPO focuses on faithful transcript reconstruction and does not directly generate restored speech audio. We leave audio-domain reconstruction to future work because fine-tuning and reinforcement learning on the speech-generation or Talker module are substantially less stable than text-side alignment in our current setting. In practice, AP-GRPO can serve as a front-end transcript reconstruction module and be followed by personalized or pathology-aware TTS systems to produce audible output.

AP-GRPO also does not fully solve the broader low-resource problem in clinical deployment. Although the reward does not require reference transcripts during RL alignment, the system still depends on a supervised fine-tuned reference policy, anchor extraction, phoneme posteriorgram quality, and disease-specific duration priors. In particular, improvements on the Parkinson's dataset are more limited, suggesting that optimization can be constrained by dataset scale, speaker diversity, and the reliability of disease-specific phonetic modeling. Larger and more diverse pathological speech corpora remain important for robust clinical generalization.

To achieve best performance, the method relies on the quality of extracted audible anchors. When anchor extraction hallucinates, produces poorly calibrated confidence scores, or misses clinically important words, the anchor gate may provide incomplete or biased structural guidance. Although monotonic matching and confidence weighting improve robustness, AP-GRPO may inherits errors from the preprocessing ASR and timestamping components.


Finally, AP-GRPO optimizes for transcript faithfulness rather than clinical correctness or communicative intent in the broadest sense. Some pathological speech contains discourse-level ambiguity, cognitive-linguistic disruption, or multiple plausible reconstructions. In such cases, a single reconstructed transcript may hide uncertainty. Practical deployment should therefore present reconstructions with uncertainty indicators, extracted anchors, and access to the original recording whenever possible.

\bibliography{reference.bib}

@inproceedings{baevski2020wav2vec2,
  title     = {wav2vec 2.0: A Framework for Self-Supervised Learning of Speech Representations},
  author    = {Baevski, Alexei and Zhou, Yuhao and Mohamed, Abdelrahman and Auli, Michael},
  booktitle = {Advances in Neural Information Processing Systems},
  volume    = {33},
  pages     = {12449--12460},
  year      = {2020},
}

@inproceedings{chen2024colmdsr,
  title     = {{CoLM-DSR}: Leveraging Neural Codec Language Modeling for Multi-Modal Dysarthric Speech Reconstruction},
  author    = {Chen, Xueyuan and Yang, Dongchao and Wang, Dingdong and Wu, Xixin and Wu, Zhiyong and Meng, Helen},
  booktitle = {Proc. Interspeech},
  pages     = {4129--4133},
  year      = {2024},
  doi       = {10.21437/Interspeech.2024-1852},
}

@inproceedings{pakinsondata,
  title     = {Parkinson’s Disease Patient Using Transfer Learning Technique},
  author    = {Yu, Qing and Ma, Yi and Li, Yongfu},
  booktitle = {Journal of Shanghai Jiaotong University},
  pages     = {1--18},
  year      = {2021},
  doi       = {https://doi.org/10.1007/s12204-021-2376-3},
}

@article{torgodata,
  title     = {Vocal Tract Representation in the Recognition of Cerebral Palsied Speech},
  author    = {Rudzicz, Frank and Hirst, Graeme and {van Lieshout}, Pascal},
  journal   = {Journal of Speech, Language, and Hearing Research},
  volume    = {55},
  number    = {4},
  pages     = {1190--1207},
  year      = {2012},
  doi       = {10.1044/1092-4388(2011/11-0223)},
}

@article{rudzicz2012torgo2,
  title     = {The {TORGO} Database of Acoustic and Articulatory Speech from Speakers with Dysarthria},
  author    = {Rudzicz, Frank and Namasivayam, Aravind Kumar and Wolff, Talya},
  journal   = {Language Resources and Evaluation},
  volume    = {46},
  number    = {4},
  pages     = {523--541},
  year      = {2012},
  doi       = {10.1007/s10579-011-9145-0},
}

@inproceedings{kim2008uaspeech,
  title     = {Dysarthric Speech Database for Universal Access Research},
  author    = {Kim, Heejin and Hasegawa-Johnson, Mark and Perlman, Adrienne and Gunderson, Jon and Huang, Thomas S. and Watkin, Kenneth and Frame, Simone},
  booktitle = {Proc. Interspeech},
  pages     = {1741--1744},
  year      = {2008},
  doi       = {10.21437/Interspeech.2008-480},
}

@article{luz2021adresso,
  title={Detecting cognitive decline using speech only: The adresso challenge},
  author={Luz, Saturnino and Haider, Fasih and De la Fuente, Sofia and Fromm, Davida and MacWhinney, Brian},
  booktitle={Proc. Interspeech},
  doi = {10.21437/Interspeech.2021-1572},
  year={2021}
}

@inproceedings{radford2023whisper,
  title={Robust speech recognition via large-scale weak supervision},
  author={Radford, Alec and Kim, Jong Wook and Xu, Tao and Brockman, Greg and McLeavey, Christine and Sutskever, Ilya},
  booktitle={International conference on machine learning},
  pages={28492--28518},
  year={2023},
  organization={PMLR}
}

@article{comanici2025gemini,
  title={Gemini 2.5: Pushing the frontier with advanced reasoning, multimodality, long context, and next generation agentic capabilities},
  author={Comanici, Gheorghe and Bieber, Eric and Schaekermann, Mike and Pasupat, Ice and Sachdeva, Noveen and Dhillon, Inderjit and Blistein, Marcel and Ram, Ori and Zhang, Dan and Rosen, Evan and others},
  journal={arXiv preprint arXiv:2507.06261},
  year={2025}
}

@article{ghosh2026audio,
  title={Audio flamingo 3: Advancing audio intelligence with fully open large audio language models},
  author={Ghosh, Sreyan and Goel, Arushi and Kim, Jaehyeon and Kumar, Sonal and Kong, Zhifeng and Lee, Sang-gil and Yang, Chao-Han and Duraiswami, Ramani and Manocha, Dinesh and Valle, Rafael and others},
  journal={Advances in Neural Information Processing Systems},
  volume={38},
  pages={41819--41886},
  year={2026}
}

@misc{xu2025qwen25omnitechnicalreport,
      title={Qwen2.5-Omni Technical Report}, 
      author={Jin Xu and Zhifang Guo and Jinzheng He and Hangrui Hu and Ting He and Shuai Bai and Keqin Chen and Jialin Wang and Yang Fan and Kai Dang and Bin Zhang and Xiong Wang and Yunfei Chu and Junyang Lin},
      year={2025},
      eprint={2503.20215},
      archivePrefix={arXiv},
      primaryClass={cs.CL},
      url={https://arxiv.org/abs/2503.20215}, 
}

@inproceedings{rombach2022latentdiffusion,
  title={High-resolution image synthesis with latent diffusion models},
  author={Rombach, Robin and Blattmann, Andreas and Lorenz, Dominik and Esser, Patrick and Ommer, Bj{\"o}rn},
  booktitle={Proceedings of the IEEE/CVF conference on computer vision and pattern recognition},
  pages={10684--10695},
  year={2022}
}

@article{bain2023whisperx,
  title={Whisperx: Time-accurate speech transcription of long-form audio},
  author={Bain, Max and Huh, Jaesung and Han, Tengda and Zisserman, Andrew},
  journal={arXiv preprint arXiv:2303.00747},
  year={2023}
}

@article{chen2025diffdsr,
  title={DiffDSR: Dysarthric Speech Reconstruction Using Latent Diffusion Model},
  author={Chen, Xueyuan and Yang, Dongchao and Wu, Wenxuan and Wu, Minglin and Xu, Jing and Wu, Xixin and Wu, Zhiyong and Meng, Helen},
  journal={arXiv preprint arXiv:2506.00350},
  year={2025}
}

@article{leung2024ttds,
  title={Training data augmentation for dysarthric automatic speech recognition by text-to-dysarthric-speech synthesis},
  author={Leung, Wing-Zin and Cross, Mattias and Ragni, Anton and Goetze, Stefan},
  journal={arXiv preprint arXiv:2406.08568},
  year={2024}
}

@article{hsu2021hubert,
  title={Hubert: Self-supervised speech representation learning by masked prediction of hidden units},
  author={Hsu, Wei-Ning and Bolte, Benjamin and Tsai, Yao-Hung Hubert and Lakhotia, Kushal and Salakhutdinov, Ruslan and Mohamed, Abdelrahman},
  journal={IEEE/ACM transactions on audio, speech, and language processing},
  volume={29},
  pages={3451--3460},
  year={2021},
  publisher={IEEE}
}

@article{la2025exploring,
  title={Exploring generative error correction for dysarthric speech recognition},
  author={La Quatra, Moreno and Koudounas, Alkis and Salerno, Valerio Mario and Siniscalchi, Sabato Marco},
  journal={arXiv preprint arXiv:2505.20163},
  year={2025}
}

@inproceedings{cuturi2017softdtw,
  title     = {Soft-{DTW}: A Differentiable Loss Function for Time-Series},
  author    = {Cuturi, Marco and Blondel, Mathieu},
  booktitle = {Proceedings of the 34th International Conference on Machine Learning},
  series    = {Proceedings of Machine Learning Research},
  volume    = {70},
  pages     = {894--903},
  year      = {2017},
  publisher = {PMLR},
}

@inproceedings{deng2025alignslm,
  title     = {Align-{SLM}: Textless Spoken Language Models with Reinforcement Learning from {AI} Feedback},
  author    = {Lin, Guan-Ting and Shivakumar, Prashanth Gurunath and Gourav, Aditya and Gu, Yile and Gandhe, Ankur and Lee, Hung-yi and Bulyko, Ivan},
  booktitle = {Proceedings of the 63rd Annual Meeting of the Association for Computational Linguistics (Volume 1: Long Papers)},
  pages     = {20395--20411},
  year      = {2025},
}

@inproceedings{hermansky2000ppg,
  title     = {Temporal Patterns ({TRAPs}) in {ASR} of Noisy Speech},
  author    = {Hermansky, Hynek and Sharma, Sangita},
  booktitle = {Proceedings of the IEEE International Conference on Acoustics, Speech, and Signal Processing (ICASSP)},
  volume    = {1},
  pages     = {289--292},
  year      = {1999},
  publisher = {IEEE},
}

@inproceedings{liu2024ppgassess,
  title     = {Deep Segmental Phonetic Posterior-Grams Based Discovery of Non-Categories in {L2} {E}nglish Speech},
  author    = {Li, Xu and Wu, Xixin and Liu, Xunying and Meng, Helen},
  booktitle = {Proceedings of the IEEE International Conference on Acoustics, Speech, and Signal Processing (ICASSP)},
  pages     = {8194--8198},
  year      = {2020},
  publisher = {IEEE},
}

@article{sakoe1978dtw,
  title   = {Dynamic Programming Algorithm Optimization for Spoken Word Recognition},
  author  = {Sakoe, Hiroaki and Chiba, Seibi},
  journal = {IEEE Transactions on Acoustics, Speech, and Signal Processing},
  volume  = {26},
  number  = {1},
  pages   = {43--49},
  year    = {1978},
  doi     = {10.1109/TASSP.1978.1163055},
}

@inproceedings{sun2016ppgvc,
  title     = {Phonetic Posteriorgrams for Many-to-One Voice Conversion without Parallel Data Training},
  author    = {Sun, Lifa and Li, Kun and Wang, Hao and Kang, Shiyin and Meng, Helen},
  booktitle = {Proceedings of the IEEE International Conference on Multimedia and Expo (ICME)},
  pages     = {1--6},
  year      = {2016},
  publisher = {IEEE},
}

@inproceedings{wang2024unitdsr,
  title     = {{UNIT-DSR}: Dysarthric Speech Reconstruction System Using Speech Unit Normalization},
  author    = {Wang, Yuejiao and Wu, Xixin and Wang, Disong and Meng, Lingwei and Meng, Helen},
  booktitle = {Proceedings of the IEEE International Conference on Acoustics, Speech, and Signal Processing (ICASSP)},
  pages     = {12306--12310},
  year      = {2024},
  publisher = {IEEE},
}

@inproceedings{wu2021ppgmispron,
  title     = {Transformer Based End-to-End Mispronunciation Detection and Diagnosis},
  author    = {Wu, Minglin and Li, Kun and Leung, Wing-Kei and Meng, Helen},
  booktitle = {Proc. Interspeech},
  pages     = {3954--3958},
  year      = {2021},
  doi       = {10.21437/Interspeech.2021-1467},
}

@inproceedings{zeng2024speechalign,
  title     = {{SpeechAlign}: Aligning Speech Generation to Human Preferences},
  author    = {Zhang, Dong and Li, Zhaowei and Li, Shimin and Zhang, Xin and Wang, Pengyu and Zhou, Yaqian and Qiu, Xipeng},
  booktitle = {Advances in Neural Information Processing Systems},
  volume    = {37},
  year      = {2024},
}

@article{song2026demma,
  title={DemMA: Dementia Multi-Turn Dialogue Agent with Expert-Guided Reasoning and Action Simulation},
  author={Song, Yutong and Wu, Jiang and Sharif, Kazi and Xu, Honghui and Dutt, Nikil and Rahmani, Amir},
  journal={arXiv preprint arXiv:2601.06373},
  year={2026}
}

@inproceedings{shao2024deepseekmath,
  title     = {{DeepSeekMath}: Pushing the Limits of Mathematical Reasoning in Open Language Models},
  author    = {Shao, Zhihong and Wang, Peiyi and Zhu, Qihao and Xu, Runxin and Song, Junxiao and Zhang, Mingchuan and Li, Y. K. and Wu, Y. and Guo, Daya},
  booktitle = {arXiv preprint arXiv:2402.03300},
  year      = {2024},
}

@article{guo2025deepseekr1,
  title   = {{DeepSeek-R1}: Incentivizing Reasoning Capability in {LLMs} via Reinforcement Learning},
  author  = {Guo, Daya and Yang, Dejian and Zhang, Haowei and Song, Junxiao and Wang, Peiyi and Zhu, Qihao and Xu, Runxin and Zhang, Ruoyu and Ma, Shirong and others},
  journal = {Nature},
  volume  = {645},
  pages   = {633--638},
  year    = {2025},
  doi     = {10.1038/s41586-025-09422-z},
}

@article{duffy2019motor,
  title     = {Motor Speech Disorders: Substrates, Differential Diagnosis, and Management},
  author    = {Duffy, Joseph R.},
  publisher = {Elsevier},
  edition   = {4th},
  year      = {2019},
}

@article{qian2023dysarthric_survey,
  title   = {A Survey of Automatic Speech Recognition for Dysarthric Speech},
  author  = {Qian, Zhaopeng and Xiao, Kejing},
  journal = {EURASIP Journal on Audio, Speech, and Music Processing},
  volume  = {2023},
  number  = {48},
  year    = {2023},
  doi     = {10.1186/s13636-023-00318-2},
}

@article{zhang2019bertscore,
  title={Bertscore: Evaluating text generation with bert},
  author={Zhang, Tianyi and Kishore, Varsha and Wu, Felix and Weinberger, Kilian Q and Artzi, Yoav},
  journal={arXiv preprint arXiv:1904.09675},
  year={2019}
}

@article{halpern2025dysarthric_overview,
  title   = {Speech Technology for Automatic Recognition and Assessment of Dysarthric Speech: An Overview},
  author  = {Halpern, Bence M. and Scharenborg, Odette and others},
  journal = {Journal of Speech, Language, and Hearing Research},
  volume  = {68},
  number  = {2},
  pages   = {547--577},
  year    = {2025},
  doi     = {10.1044/2024_JSLHR-23-00740},
}

\newpage

\appendix
\section*{Appendix}
\label{sec:appendix}

\section{Implementation Details}
\label{app:implementation}

\subsection{Anchor Extraction and Matching}
\label{app:anchor_extraction}

Anchors are obtained by a two-pass pipeline. First, Whisper large-v3 provides candidate words and confidence scores; words above confidence threshold $\tau=0.85$ and with at least three characters are retained. Second, WhisperX provides word-level timestamps through forced alignment. The final anchor confidence is
\[
c_m=\min(c_m^{\mathrm{Whisper}},c_m^{\mathrm{WhisperX}}).
\]

For each candidate transcript, anchors are matched using normalized monotonic matching. We lowercase text, remove punctuation, and tokenize both anchors and candidate transcripts. Anchors are processed from left to right. For anchor $a_m$, we search only after the previous matched candidate position and select the first candidate token span whose normalized string similarity to $a_m$ exceeds threshold $\rho$. If no such span exists, $a_m$ is marked unmatched. This rule handles repeated anchors by assigning each anchor to the earliest feasible unmatched position and ensures that matched anchors preserve the temporal order of the original speech signal.

\subsection{Model Details}

\paragraph{Backbone.}
We use Qwen2.5-Omni-7B and Audio-Flamingo3 with the Talker module disabled. LoRA is applied to attention layers with rank 16 and alpha 32. The audio encoder is frozen.

\paragraph{SFT.}
The supervised fine-tuning learning rate is $2\times10^{-5}$. We use batch size 3 with gradient accumulation 4. We conduct SFT for 2 epoches.

\paragraph{GRPO.}
We use a single GRPO round with $T=N/3*mini\_batchsize$ steps, group size $G=32$, temperature $1.2$, top-$p=0.9$, maximum output length 64 tokens, clipping parameter $\epsilon=0.2$, and learning rate $5\times10^{-6}$. The dual update uses base step size $\beta_\mu=0.05$, momentum decay $\gamma=0.9$, boost factor $c=2.0$, initial multiplier $\mu_{\mathrm{init}}=1.0$, maximum multiplier $\mu_{\max}=10.0$, and anchor threshold $\alpha=0.95$. The KL coefficient is $\eta_0=0.005$. The hard anchor gate uses $\epsilon_g=0.05$.

\subsection{Phonetic scoring}
We use wav2vec2-lv-60-espeak-cv-ft as a frozen phoneme posteriorgram extractor. The confusion matrix uses articulatory groups with $C_{\mathrm{group}}=0.4$ and $C_{\mathrm{default}}=0.1$. Soft-DTW uses smoothing parameter $\gamma=0.1$. Span expansion uses $\eta=0.5$ and includes anchor phonemes. We require a minimum span length of 5 frames and cap span length at 500 frames.

\subsection{TTDS}
TTDS is trained separately for each disease condition and is used for G2P conversion and dysarthric duration prediction. We perform a full-sentence forward pass and slice the resulting phoneme-duration sequence by span. The maximum duration per phoneme is capped at 10 frames. All parameters follow the original article.

\subsection{Hardware}
Experiments are conducted on $3\times$ NVIDIA A100-80GB GPUs.

\section{Preliminary: Phoneme Posteriorgrams and Soft-DTW for Phonetic Similarity}
\label{app:ppg_softdtw_background}

This section provides background on phoneme posteriorgrams (PPGs) and Soft-DTW, the two components that underlie the inter-anchor phonetic reward in AP-GRPO. 


\subsection{Phoneme Posteriorgrams (PPGs)}

A phoneme posteriorgram is a time--phoneme matrix that summarizes the frame-level phonetic content of a speech signal. Given a speech waveform $x$ of duration $T_F$ frames, a phoneme recognition model $\phi$ produces
\[
P = \phi(x) \in \mathbb{R}^{T_F \times K},
\]
where $K$ is the size of the phoneme inventory and $P(t,k)$ is the posterior probability that phoneme $k$ is active at frame $t$. Each row $P(t,\cdot)$ is a probability distribution over phonemes, and each column $P(\cdot,k)$ traces the temporal activation of phoneme $k$ across the utterance.

PPGs were originally introduced for speaker-independent speech recognition~\citep{hermansky2000ppg} and have since been adopted for voice conversion~\citep{sun2016ppgvc}, where they serve as a speaker-independent intermediate representation that captures \emph{what is said} while discarding speaker identity. More recently, PPGs have been used for pronunciation assessment~\citep{liu2024ppgassess}, mispronunciation detection~\citep{wu2021ppgmispron}, and cross-lingual speech analysis, in each case exploiting the property that PPGs encode phonetic content in a form that is invariant to speaker, channel, and (to some extent) language.

In AP-GRPO, we extract PPGs using a pretrained wav2vec2-based phoneme CTC model~\citep{baevski2020wav2vec2}. Specifically, we use wav2vec2-lv-60-espeak-cv-ft, which was trained on multilingual speech with an espeak-based phoneme inventory. For a speech span $x^{(n)}$ between two audible anchors, the model produces
\[
P^{(n)} = \phi(x^{(n)}) \in \mathbb{R}^{T_n \times K}.
\]
This PPG encodes what phonemes the patient produced in the inter-anchor region, even if the pronunciation is distorted. The model is frozen during all AP-GRPO training; it is used only as a phonetic feature extractor.


\subsection{Dynamic Time Warping (DTW)}

Given two sequences of potentially different lengths, dynamic time warping finds an optimal monotonic alignment between them. Let $\mathbf{a} = (a_1, \ldots, a_T)$ and $\mathbf{b} = (b_1, \ldots, b_L)$ be two sequences, and let $\Delta(a_t, b_l)$ be a local cost function measuring the dissimilarity between elements $a_t$ and $b_l$. DTW computes
\[
\mathrm{DTW}(\mathbf{a}, \mathbf{b})
=
\min_{\pi \in \mathcal{A}_{T,L}}
\sum_{(i,j) \in \pi}
\Delta(a_i, b_j),
\]
where $\mathcal{A}_{T,L}$ is the set of all monotonic alignment paths from $(1,1)$ to $(T,L)$. Each alignment path $\pi$ is a sequence of index pairs $(i,j)$ satisfying monotonicity (indices only increase) and boundary constraints (the path starts at $(1,1)$ and ends at $(T,L)$). The DTW distance is the minimum total cost over all valid paths.

DTW has been widely used in speech processing for template-based recognition~\citep{sakoe1978dtw}, keyword spotting, and pronunciation scoring. Its key property is tolerance to temporal variation: two realizations of the same phoneme sequence that differ in speaking rate will still align well under DTW, because the warping path can stretch or compress either sequence locally.

However, standard DTW has two limitations for use in a reward function. First, it is not differentiable with respect to continuous inputs, because the $\min$ operation over discrete alignment paths produces a piecewise-constant objective. Second, it is sensitive to the single best alignment path, which can be unstable when the cost matrix is noisy.

\subsection{Soft-DTW}

Soft-DTW~\citep{cuturi2017softdtw} addresses both limitations by replacing the hard minimum over alignment paths with a soft minimum:
\[
\mathrm{SoftDTW}_\gamma(\mathbf{a}, \mathbf{b})
=
\mathrm{softmin}_{\pi \in \mathcal{A}_{T,L}}^{\gamma}
\sum_{(i,j) \in \pi}
\Delta(a_i, b_j),
\]
where the soft-minimum operator is defined as
\[
\mathrm{softmin}^{\gamma}(z_1, \ldots, z_N)
=
-\gamma \log \sum_{n=1}^{N} \exp(-z_n / \gamma).
\]
As $\gamma \to 0$, $\mathrm{SoftDTW}_\gamma$ recovers standard DTW. For $\gamma > 0$, Soft-DTW is differentiable and averages over exponentially many alignment paths, weighted by their cost. Paths with lower total cost receive higher weight.

In practice, Soft-DTW is computed via a Bellman recursion analogous to standard DTW. Let $D \in \mathbb{R}^{T \times L}$ be the cost matrix with $D(t,l) = \Delta(a_t, b_l)$. Define
\[
R(0,0) = 0,
\]

\[\qquad R(i,j) = +\infty \text{ for } i < 0 \text{ or } j < 0,\]
and for $i = 1, \ldots, T$ and $j = 1, \ldots, L$:
\begin{equation}
\begin{aligned}
R(i,j)
& =
D(i,j)
+ \mathrm{softmin}^{\gamma}
\bigl(
R(i{-}1,j),\; \\
& R(i,j{-}1),\;
R(i{-}1,j{-}1)
\bigr).
\label{eq:app_softdtw_recursion}
\end{aligned}
\end{equation}
The Soft-DTW value is $\mathrm{SoftDTW}_\gamma(D) = R(T, L)$. The computation is $O(T \times L)$, the same as standard DTW.

The temperature parameter $\gamma$ controls the sharpness of the alignment. Small $\gamma$ (e.g., $0.01$) concentrates weight on the single best path, approaching hard DTW. Large $\gamma$ (e.g., $1.0$) distributes weight more uniformly across paths, producing a smoother but less discriminative score. In AP-GRPO, we use $\gamma = 0.1$, which provides a smooth score while remaining sensitive to alignment quality.

\subsection{Measuring Phonetic Similarity with PPG + Soft-DTW}

Combining PPGs with Soft-DTW yields a phonetic similarity measure between a speech signal and a candidate phoneme sequence. The setup is:

\begin{itemize}
\item \textbf{Speech side:} A speech span $x^{(n)}$ is converted to a PPG $P^{(n)} \in \mathbb{R}^{T \times K}$ by the frozen wav2vec2 model. Each row $P^{(n)}(t, \cdot)$ is a distribution over phonemes at frame $t$.
\item \textbf{Text side:} A candidate text span is converted to a phoneme sequence $(q_1, \ldots, q_L)$ via G2P. Optionally, each phoneme is expanded by its predicted duration to produce a frame-level phoneme path $(\tilde{q}_1, \ldots, \tilde{q}_{\tilde{L}})$.
\item \textbf{Cost matrix:} The element $D(t, l)$ measures how well frame $t$ of the speech supports the candidate phoneme at position $l$:
\[
D(t, l) = -\log\bigl(P^{(n)}(t, \tilde{q}_l) + \epsilon\bigr).
\]
When the PPG assigns high probability to the candidate phoneme $\tilde{q}_l$ at frame $t$, the cost is low.
\item \textbf{Alignment score:} Soft-DTW finds the minimum-cost monotonic alignment between the $T$ speech frames and the $\tilde{L}$ candidate phoneme positions. The normalized score is
\[
S = -\frac{\mathrm{SoftDTW}_\gamma(D)}{T + \tilde{L}},
\]
where normalization by $T + \tilde{L}$ prevents systematic bias toward shorter sequences. Higher $S$ indicates stronger phonetic compatibility.
\end{itemize}

This combination has several properties that make it suitable for scoring pathological speech reconstructions:

\paragraph{Temporal flexibility.} Soft-DTW tolerates speaking rate variation, which is common in dysarthria (e.g., Parkinson's bradykinesia, ALS fatigue-related slowing). A candidate phoneme sequence that is phonetically correct but temporally compressed or expanded relative to the speech signal will still achieve a good alignment score, because the warping path compensates for rate differences.

\paragraph{Soft phonetic matching.} The PPG provides a soft distribution over phonemes at each frame, rather than a hard phoneme label. When a patient produces a distorted phoneme, the PPG typically assigns partial probability mass to the intended phoneme and to acoustically similar neighbors. This makes the score tolerant of mild articulatory imprecision without requiring explicit knowledge of the patient's error patterns. In AP-GRPO, we further improve this tolerance through confusion smoothing (Section~\ref{eq:confusion_smooth} in the main text).

\paragraph{Monotonic alignment.} The DTW constraint that the alignment path must be monotonic is well suited to speech, where phonemes are produced in a fixed temporal order. Unlike bag-of-phoneme comparisons, Soft-DTW rewards candidates whose phonemes appear in the correct temporal sequence relative to the speech signal.

\paragraph{Computational cost.} The Soft-DTW recursion is $O(T \times L)$ per span. For a typical inter-anchor span with $T \approx 50$ frames and $L \approx 30$ expanded phoneme positions, the computation takes under 1~ms on GPU. This is negligible compared to the cost of candidate generation from the speech LLM.

\section{Confusion Matrix Construction}
\label{app:confusion}

The phoneme confusion matrix $C\in[0,1]^{K\times K}$ is constructed from articulatory feature groups. The diagonal entries are set to $C(q,q)=1.0$. Phoneme pairs within the same articulatory group receive $C(q,k)=0.4$, and unrelated pairs receive $C(q,k)=0.1$.

The default groups are based on common dysarthric confusions:
\[
\begin{aligned}
&\text{front unrounded vowels: } /i,\ IH,\ e,\ EH/,\\
&\text{back rounded vowels: } /u,\ UH,\ o,\ AO/,\\
&\text{open vowels: } /a,\ AE,\ AA,\ AH/,\\
&\text{central vowels: } /AX,\ ER/,\\
&\text{bilabial stops: } /p,\ b/,\\
&\text{alveolar stops: } /t,\ d/,\\
&\text{velar stops: } /k,\ g/,\\
&\text{sibilant fricatives: } /s,\ z,\ SH,\ ZH/,\\
&\text{non-sibilant fricatives: } /f,\ v,\ TH,\ DH/,\\
&\text{nasals: } /m,\ n,\ NG/,\\
&\text{liquids: } /l,\ R/,\\
&\text{glides: } /w,\ y/.
\end{aligned}
\]
The matrix can also be refined using empirical phoneme-confusion statistics from a dysarthric development set.

\section{Primal Oracle Validation}
\label{app:epsilon_p}

\begin{figure}[t]
    \centering
    \includegraphics[width=\linewidth]{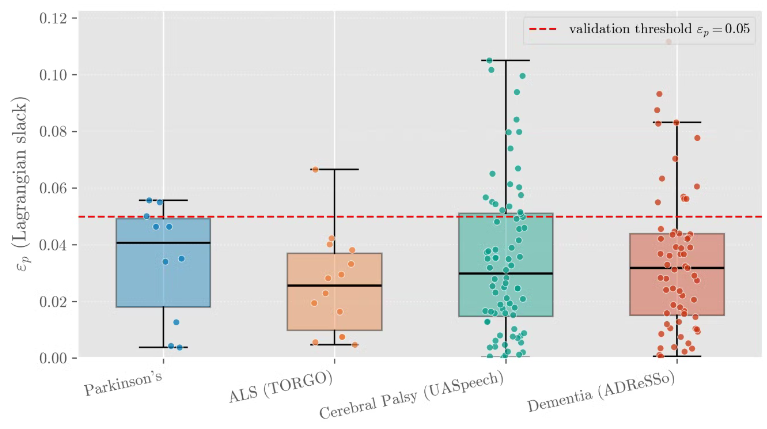}
    \caption{Lagrangian slack $\epsilon_p$ statistics on each GRPO steps
}
    \label{Lagrangian_slack}
\end{figure}

Assumption~\ref{assum:primal} requires that each GRPO update approximately maximizes the Lagrangian for the current multiplier. We validate this assumption by estimating $\epsilon_p$ through random LoRA weight perturbation.

At every fifth training step, we perturb the LoRA weights in 20 random directions with perturbation scale $0.01$ and recompute the Lagrangian for each perturbed model. We estimate
\[
\epsilon_p
=
\max
\left(
0,\;
\max_k
\Lcal(\theta+\delta_k,\mu_t)
-
\Lcal(\theta_t,\mu_t)
\right).
\]
We report the median and 95th percentile of $\epsilon_p$ across all measurement points at Figure~\ref{Lagrangian_slack}. In our validation criterion, the median $\epsilon_p$ should remain below $0.05$.

\section{Datasets Details}
\label{app:data_stats}

We evaluate AP-GRPO on four pathological speech datasets covering neurodegenerative and neuromotor disorders. All datasets were originally collected and released by their respective authors and institutions for research use. Our study uses these existing datasets only; we do not recruit new participants or collect new speech recordings. Where provided by the dataset documentation, data collection was conducted under institutional ethics or IRB approval, and all recordings are used according to the corresponding dataset access terms. The dataset statistics are presented at Table~\ref{tab:datasets}

\paragraph{ADReSSo.}
ADReSSo is an Alzheimer's and dementia speech benchmark derived from spontaneous picture-description recordings. It is designed for acoustic-only Alzheimer's dementia recognition and provides speech recordings balanced for demographic factors such as age and gender. In our setting, ADReSSo represents the linguistically fragmented side of pathological speech: the acoustic signal may remain partly intelligible, while lexical retrieval, semantic consistency, and discourse organization can be impaired. We use ADReSSo to evaluate whether AP-GRPO can preserve reliable audible anchors while reconstructing missing or degraded speech content in dementia-related speech.

\paragraph{TORGO.}
TORGO is a dysarthric speech corpus collected by the University of Toronto in collaboration with clinical partners. It contains aligned acoustic recordings, transcripts, and articulatory information from speakers with dysarthria caused by cerebral palsy or amyotrophic lateral sclerosis, together with non-dysarthric controls. In this work, we use the dysarthric ALS-related subset to evaluate reconstruction under severe motor-speech degradation. TORGO is especially useful for AP-GRPO because articulatory impairment often produces sparse but reliable audible anchors surrounded by acoustically distorted inter-anchor regions.

\paragraph{UASpeech.}
UASpeech is a dysarthric speech dataset collected at the University of Illinois Urbana-Champaign from speakers with cerebral palsy. The corpus contains isolated-word recordings with multiple repetitions across common words, uncommon words, digits, computer commands, and radio alphabet tokens. The official dataset documentation states that collection was approved by the University of Illinois Institutional Review Board. We use UASpeech to evaluate anchor preservation and phonetic reconstruction under cerebral-palsy-related dysarthria, where atypical articulation and reduced intelligibility can strongly affect ASR and transcript reconstruction.

\paragraph{SJTU-PD.}
SJTU-PD is a Parkinson's disease speech dataset released by the SJTU research group. It contains speech recordings from ten patients diagnosed with Parkinson's disease, including original and denoised speech versions. We use SJTU-PD to evaluate AP-GRPO under hypokinetic dysarthria, where reduced loudness, imprecise articulation, altered prosody, and abnormal speaking rate can affect both anchor extraction and inter-anchor phonetic alignment. This dataset complements TORGO and UASpeech by testing the method on neurodegenerative motor-speech impairment rather than primarily developmental or neuromotor dysarthria.

\begin{table*}[t]
\centering
\caption{Dataset summary.}
\label{tab:datasets}
\small
\begin{tabular}{@{}llccl@{}}
\toprule
Dataset & Disease & Speakers & Paired Speech-Text data & Source \\
\midrule
ADReSSo & Alzheimer's/Dementia &  & Real & \citet{luz2021adresso} \\
TORGO & ALS & 8 dysarthric & Real \& Lexical Aligned & U.\ Toronto \\
UASpeech & Cerebral Palsy & 15 dysarthric & Real \& Lexical Aligned & UIUC \\
SJTU-PD & Parkinson's & 10 patients & Real & SJTU \\
\bottomrule
\end{tabular}
\end{table*}

\subsection{Lexical Align Preprocessing for Word-Level Corpora}
\label{app:lexical_merge}

\begin{table*}[t]
\centering
\caption{Corpus statistics before and after lexical align.}
\label{tab:lexical_merge}
\small
\begin{tabular}{@{}llcccc@{}}
\toprule
Corpus & Stage & Speakers & Utterances & Mean words/utt. & Vocab coverage \\
\midrule
\multirow{3}{*}{TORGO}
 & Original (words) & 15 & $\sim$7{,}800 & 1.0 & --- \\
 & After Stage 1 & 15 & 2{,}470 & 3.3 & 100\% \\
 & After Stage 2 & 15 & 833 & 10.4 & 100\% \\
\midrule
\multirow{2}{*}{UASpeech}
 & Original (words) & 25 & $\sim$9{,}500 & 1.0 & --- \\
 & After Stage 2 & 25 & 3{,}819 & 9.5 & 100\% \\
\bottomrule
\end{tabular}
\end{table*}

Several dysarthric speech corpora, including TORGO and UASpeech, are collected at the single-word level: each recording contains one isolated word produced in response to a prompt. This word-level format is unsuitable for sentence-level speech reconstruction, where the model must recover multi-word utterances from continuous pathological speech. We therefore apply a two-stage \emph{lexical align} procedure that concatenates word-level recordings into pseudo-sentences while preserving per-speaker vocabulary coverage.

\subsubsection{Stage 1: Template-Driven Sentence Construction}

For each speaker, we construct short pseudo-sentences (target $\sim$3--4 words) from the speaker's single-word recordings.

\paragraph{POS tagging.}
We POS-tag the global single-word vocabulary using NLTK, with two accuracy improvements: each word is tagged in a pseudo-context (``the \textsc{word}'') to provide the tagger with a syntactic anchor, and high-frequency function words (determiners, prepositions, conjunctions, pronouns, copulas) are overridden with a hand-coded dictionary. Penn Treebank tags are bucketed into nine coarse classes: \textsc{noun}, \textsc{verb}, \textsc{adj}, \textsc{adv}, \textsc{det}, \textsc{prep}, \textsc{pron}, \textsc{conj}, and \textsc{other}.

\paragraph{Template-driven generation.}
We define 22 English sentence templates ordered from short to long, such as [\textsc{det}, \textsc{noun}, \textsc{verb}], [\textsc{det}, \textsc{adj}, \textsc{noun}], [\textsc{det}, \textsc{noun}, \textsc{verb}, \textsc{det}, \textsc{noun}], and [\textsc{noun}, \textsc{prep}, \textsc{det}, \textsc{noun}]. For each speaker, we build a speaker-local vocabulary indexed by POS class and iterate over all templates for up to 8 passes. In each pass, template slots are filled by preferring unused words when available and falling back to previously used words otherwise. Templates that cannot be filled due to missing POS classes in the speaker's vocabulary are skipped. Generation stops when per-speaker vocabulary coverage reaches 90\%.

\paragraph{Tail coverage.}
Any words still unused after template generation are grouped into chunks of 4 and emitted as verbatim filler sequences, guaranteeing 100\% per-speaker vocabulary coverage.

\paragraph{Audio concatenation.}
For each pseudo-sentence, the constituent word-level recordings are concatenated with 80\,ms of silence between words. We select one canonical recording per (speaker, word) pair, preferring the array microphone channel when available.

\subsubsection{Stage 2: Sentence Lengthening}

Stage~1 produces short pseudo-sentences (mean $\sim$3.3 words). To better approximate natural sentence lengths, Stage~2 greedily stitches Stage~1 outputs into longer sequences.

For each speaker, Stage~1 sentences are sorted by word count in descending order. A greedy algorithm pops sentences from the queue and appends them until the accumulated word count exceeds a target length (default: 9 words for TORGO, 8 words for UASpeech), with a hard minimum floor (6 words for TORGO, 5 words for UASpeech). Audio segments are concatenated with a longer inter-sentence pause (160\,ms for TORGO, 250\,ms for UASpeech) to acoustically distinguish sentence boundaries from word boundaries. Residual sentences that cannot independently reach the minimum floor are folded into the last emitted sentence.

\subsubsection{Corpus Statistics After Lexical Align}

The merged corpora retain 100\% per-speaker vocabulary coverage while providing utterance lengths comparable to natural conversational speech. All AP-GRPO experiments on TORGO and UASpeech use the Stage~2 merged data.

\section{Anchor Coverage Statistics}
\label{app:anchor_stats}

Before training, we analyze the distribution of anchor coverage across datasets and confidence thresholds to understand the raw material available to AP-GRPO. Figure~\ref{fig:anchor_threshold_stats} reports anchor coverage (number of anchors divided by number of words in the reference transcript) per utterance, stratified by utterance length and Whisper confidence threshold $\tau \in \{0.80, 0.85, 0.90, 0.95\}$.


We make four observations.

\paragraph{Coverage decreases sharply with $\tau$.}
At $\tau = 0.80$, median coverage ranges from 20--50\% across datasets and utterance lengths, providing a reasonable density of audible anchors for defining inter-anchor spans. At $\tau = 0.90$, coverage drops to 5-15\% for most conditions - roughly one or two anchors per 10-word utterance. At $\tau = 0.95$, coverage is near zero for all datasets except TORGO's original single-word recordings: almost no words survive the confidence threshold.

This has direct consequences for AP-GRPO. With very few anchors, inter-anchor spans stretch across the entire utterance, and the Soft-DTW cost matrix becomes too large and underconstrained for discriminative alignment. This explains the training stall observed in early experiments with $\tau = 0.9$: the phonetic reward produced near-identical scores across all candidates because the spans were too long to differentiate. Our final setting of $\tau = 0.85$ provides denser anchors and shorter, more discriminative spans.


\begin{figure*}[t]
\centering

\begin{subfigure}{0.235\textwidth}
\centering
\includegraphics[width=\linewidth]{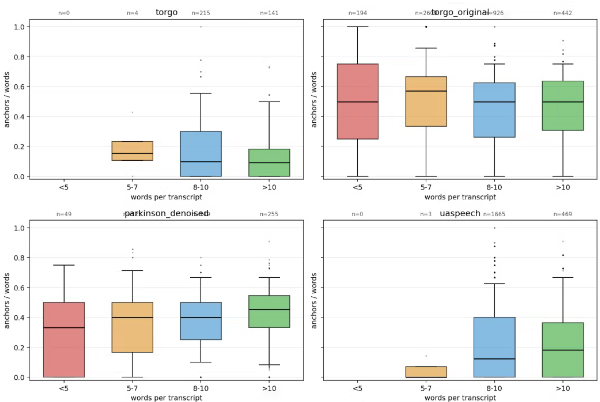}
\caption{$\tau=0.8$}
\end{subfigure}
\hfill
\begin{subfigure}{0.235\textwidth}
\centering
\includegraphics[width=\linewidth]{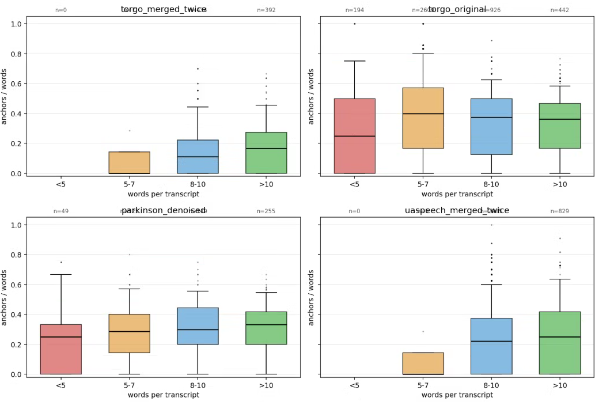}
\caption{$\tau=0.85$}
\end{subfigure}
\hfill
\begin{subfigure}{0.235\textwidth}
\centering
\includegraphics[width=\linewidth]{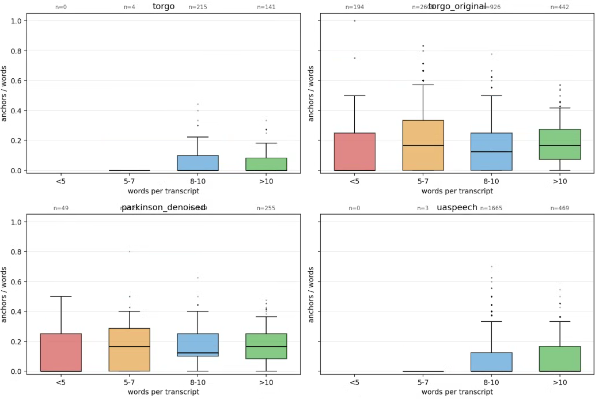}
\caption{$\tau=0.9$}
\end{subfigure}
\hfill
\begin{subfigure}{0.235\textwidth}
\centering
\includegraphics[width=\linewidth]{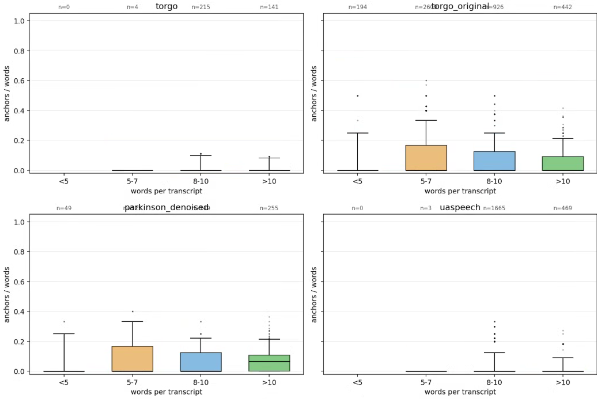}
\caption{$\tau=0.95$}
\end{subfigure}

\caption{Anchor Coverage per transcript, by word count bucket, for different Threshold.}
\label{fig:anchor_threshold_stats}
\end{figure*}

\paragraph{Coverage varies with utterance length.}
Shorter utterances ($< 5$ words) exhibit high variance: some have a single anchor word covering 100\% of content, while others have none. Longer utterances ($> 10$ words) show more stable but generally lower median coverage, because the additional words tend to include less intelligible function words and connecting phrases that fall below the confidence threshold. The most informative regime for AP-GRPO is medium-length utterances (5--10 words), where anchors are frequent enough to define spans but enough corrupted content remains for the phonetic reward to differentiate candidates.

\paragraph{Disease-specific coverage profiles.}
Even at the same $\tau$, the four datasets show distinct patterns that reflect each condition's impact on speech intelligibility:
\begin{itemize}
    \item \textbf{TORGO merged (ALS):} The lowest coverage across all thresholds. ALS causes progressive motor neuron loss, severely degrading articulatory precision. At $\tau = 0.85$, median coverage is below 20\% for utterances longer than 5 words, and at $\tau = 0.90$, coverage is nearly zero. This is the hardest condition for AP-GRPO and predicts the highest $\mu^*$.
    \item \textbf{TORGO original (ALS):} Substantially higher coverage than TORGO merged, demonstrating that the merging procedure---not just the disease---accounts for much of the anchor sparsity. This comparison isolates the effect of lexical align on reconstruction difficulty.
    \item \textbf{Parkinson's (SJTU-PD):} Moderate coverage at $\tau = 0.80$ (30--50\%), dropping to 10--15\% at $\tau = 0.90$. Parkinson's causes hypokinetic dysarthria with reduced loudness and imprecise consonants, but many words remain recognizable. The relatively stable coverage across word-count buckets suggests that the intelligibility degradation is uniform rather than concentrated on specific word types.
    \item \textbf{UASpeech merged (cerebral palsy):} Similar to Parkinson's in overall coverage level, but with a sharper drop between $\tau = 0.80$ and $\tau = 0.85$. This suggests that many words hover near the 0.80-0.85 confidence boundary: they are partially recognizable but not reliably so. The choice of $\tau$ is particularly consequential for this dataset.
\end{itemize}

\paragraph{Implications for $\tau$ selection.}
These statistics motivate our choice of $\tau = 0.85$ for the main experiments. A threshold of $\tau = 0.90$ or above leaves most datasets with too few anchors for the phonetic reward to operate (the ``anchor desert'' problem), while $\tau \leq 0.5$ would admit noisy anchors that could mislead the anchor gate. The $\tau$ sensitivity analysis in Figure~\ref{fig:learning_curve} (c) confirms that $\tau = 0.85$ provides a robust operating point across conditions.






\section{Constrained Optimization}
\label{sec:optimization}

\subsection{Formulation}

AP-GRPO constrains anchor coverage while maximizing the log-space reward:
\begin{align}
    \max_\theta \quad
    &
    \E_{\pitheta}
    \bigl[
    \log \ganchor+\Rphon
    \bigr],
    \label{eq:objective}
    \\
    \text{s.t.}\quad
    &
    \E_{\pitheta}[\ganchor]\geq\alpha,
    \\
    &
    \E_x[\KL(\pitheta\|\piref)]\leq\delta.
    \label{eq:constraint}
\end{align}
The corresponding Lagrangian introduces a dual variable $\mu\geq0$:

\begin{equation}
\begin{aligned}
    \Lcal(\theta,\mu)
    & =
    \E
    \left[
    \log \ganchor+\Rphon+\mu\ganchor
    \right]\\
    & -
    \mu\alpha
    -
    \eta_t\KL(\pitheta\|\piref).
\end{aligned}
\end{equation}

The anchor gate enters optimization in two ways. The $\log \ganchor$ term provides a persistent anchor-derived reward signal. The multiplier term $\mu\ganchor$ adaptively enforces anchor coverage when the constraint is violated. When the constraint is satisfied and $\mu$ decays toward zero, the log-gate term continues to provide a nonzero coverage gradient.

\subsection{GRPO Advantage}

The per-candidate reward $R_j=\log g_j+g_{\mathrm{anchor},j}$ 

and the constraint signal $g_j$ have different variance profiles. We therefore normalize them independently within the sampled group:

\[
\hat R_{\mathrm{norm},j}
=
\frac{R_j-\bar R}{\hat\sigma_R+\epsilon},
\qquad
\hat g_{\mathrm{norm},j}
=
\frac{g_j-\bar g}{\hat\sigma_g+\epsilon}.
\]
The advantage for candidate transcript $y_j$ is
\begin{equation}
    A_j
    =
    \hat R_{\mathrm{norm},j}
    +
    \mu_t\hat g_{\mathrm{norm},j}.
\end{equation}
Decoupled normalization prevents the high-variance phonetic score from obscuring the lower-variance anchor-coverage signal.

\subsection{Momentum-Accelerated Dual Update}

The dual variable is updated by projected subgradient descent with momentum. Let $\hat g_t$ denote the batch-level estimate of anchor coverage. We maintain an exponential moving average of constraint violations:
\begin{align}
    \tilde v_{t+1}
    &=
    \gamma \tilde v_t
    +
    (1-\gamma)(\alpha-\hat g_t),
    \label{eq:ema}
    \\
    \beta_t^{\mathrm{eff}}
    &=
    \beta_\mu
    \left(
    1+c[\tilde v_{t+1}]^+
    \right),
    \label{eq:beta_eff}
    \\
    \mu_{t+1}
    &=
    \left[
    \mu_t
    +
    \beta_t^{\mathrm{eff}}
    (\alpha-\hat g_t)
    \right]^+.
    \label{eq:dual_update}
\end{align}
When violations persist, $\tilde v_{t+1}>0$ and the effective step size increases up to $\beta_\mu(1+c)$. When the constraint is satisfied, the update returns to the base step size.

\subsection{Scheduled KL Penalty}

The KL coefficient follows a linear weak-to-strong schedule:
\[
\eta_t
=
\eta_0
+
(\eta_T-\eta_0)t/T,
\qquad
\eta_0<\eta_T.
\]

\section{Theoretical Analysis}
\label{app:proofs}

\subsection{Assumptions}

Throughout the appendix, $x$ denotes a pathological speech recording, $y$ denotes a candidate reconstruction transcript, and $\pitheta(\cdot|x)$ denotes the speech LLM policy. The anchor gate is denoted by $\ganchor(y,\Acal)$, and the inter-anchor phonetic reward is denoted by $\Rphon(y,x,\Acal)$.

\begin{assumption}[Bounded rewards]
\label{assum:bounded}
The anchor gate satisfies
\[
\ganchor(y,\Acal)\in[\epsilon_g,1],
\qquad
\epsilon_g>0.
\]
The inter-anchor phonetic score satisfies
\[
|\Rphon(y,x,\Acal)|\leq B.
\]
Consequently, the log-space reward
\[
R(y,x,\Acal)=\log \ganchor(y,\Acal)+\Rphon(y,x,\Acal)
\]
is bounded:
\[
R(y,x,\Acal)\in[\log\epsilon_g-B,\;B].
\]
\end{assumption}

\begin{assumption}[Feasibility]
\label{assum:feasibility}
There exists a policy $\pi^*\in\Pi_\delta$ satisfying the anchor-coverage constraint:
\[
\E_{\pi^*}[\ganchor]\geq \alpha,
\]
where
\[
\Pi_\delta
:=
\left\{
\pitheta:
\E_x[\KL(\pitheta\|\piref)]\leq\delta
\right\}.
\]
\end{assumption}

\begin{assumption}[Approximate primal oracle]
\label{assum:primal}
At each dual step $t$, GRPO produces a policy $\pi_{\theta_t}$ satisfying
\[
\Lcal(\theta_t,\mu_t)
\geq
\max_{\theta\in\Theta_\delta}
\Lcal(\theta,\mu_t)
-
\epsilon_p.
\]
The approximation error $\epsilon_p$ is validated empirically in Section~\ref{app:epsilon_p}. For the constraint-violation bound to be informative, we require $\epsilon_p\ll \mu^*$.
\end{assumption}

\begin{assumption}[Independent group sampling]
\label{assum:iid}
For each speech recording $x$, the $G$ candidate transcripts are sampled independently from $\pi_{\theta_t}(\cdot|x)$.
\end{assumption}

\begin{assumption}[Bounded gate variance]
\label{assum:variance}
The within-group variance of the anchor gate is uniformly bounded:
\[
\mathrm{Var}_{y\sim\pi_{\theta_t}}
\left(
\ganchor(y,\Acal)
\right)
\leq
\sigma^2.
\]
\end{assumption}

\subsection{Lagrangian Equivalence}

The constrained optimization problem in Eqs.~\eqref{eq:objective}--\eqref{eq:constraint} maximizes the log-space reward under an anchor-coverage constraint and a KL budget. Its Lagrangian is
\begin{equation}
\begin{aligned}
    \Lcal(\theta,\mu)
    & =
    \E
    \left[
    \log \ganchor+\Rphon+\mu\ganchor
    \right]\\
    & -
    \mu\alpha
    -
    \eta_t\KL(\pitheta\|\piref).
\end{aligned}
\end{equation}
where $\mu\geq0$ is the Lagrange multiplier for the anchor-coverage constraint.

\begin{theorem}[Lagrangian equivalence]
\label{thm:lagrangian_formal}
AP-GRPO with adaptive multiplier $\mu_t$ performs approximate primal-dual optimization of Eqs.~\eqref{eq:objective}--\eqref{eq:constraint}.

\begin{enumerate}
    \item The GRPO advantage
    \[
    A_j=\hat R_{\mathrm{norm},j}+\mu_t\hat g_{\mathrm{norm},j}
    \]
    in Eq.~\eqref{eq:advantage} is a clipped surrogate for the policy-gradient direction of $\Lcal(\theta,\mu_t)$.

    \item The dual update
    \[
    \mu_{t+1}
    =
    \left[
    \mu_t
    +
    \beta_t^{\mathrm{eff}}
    (\alpha-\hat g_t)
    \right]^+
    \]
    is stochastic projected subgradient descent on the convex dual function
    \[
    d(\mu)
    =
    \max_{\theta\in\Theta_\delta}
    \Lcal(\theta,\mu).
    \]
\end{enumerate}
\end{theorem}

\begin{proof}
For the primal update, the reward $R=\log \ganchor+\Rphon$ and the constraint signal $\ganchor$ enter the Lagrangian as
\[
\log \ganchor+\Rphon+\mu_t\ganchor.
\]
The AP-GRPO advantage uses a decoupled normalization of $R$ and $\ganchor$, followed by the same effective weighting induced by the Lagrangian. Therefore, the clipped GRPO objective approximates the policy-gradient direction of $\Lcal(\theta,\mu_t)$, up to sampling, clipping, and primal approximation error.

For the dual update, the dual function $d(\mu)$ is convex because it is the pointwise maximum of functions affine in $\mu$. Let $\pi^*_\mu$ denote a Lagrangian-maximizing policy for fixed $\mu$. A subgradient of $d$ is
\[
\partial_\mu d(\mu)
=
\E_{\pi^*_\mu}[\ganchor]
-
\alpha.
\]
Projected subgradient descent on $\min_{\mu\geq0}d(\mu)$ gives
\begin{equation}
\begin{aligned}
\mu_{t+1} & =
\left[
\mu_t
-
\beta
\left(
\E_{\pi^*_\mu}[\ganchor]-\alpha
\right)
\right]^+ \\
& =
\left[
\mu_t
+
\beta
\left(
\alpha-\E_{\pi^*_\mu}[\ganchor]
\right)
\right]^+.
    \end{aligned}
\end{equation}

Replacing exact expectations by batch estimates and the exact maximizer by the GRPO policy $\pi_{\theta_t}$ yields the stochastic projected update used in AP-GRPO. Momentum acceleration only replaces the fixed step size $\beta$ with a bounded time-varying step size $\beta_t^{\mathrm{eff}}$.
\end{proof}

\begin{remark}[Log reward and signal persistence]
\label{rem:log_signal}
The term $\log \ganchor$ provides a persistent anchor-derived optimization signal without an auxiliary baseline weight. Since
\[
\frac{\partial \log g}{\partial g}
=
\frac{1}{g}
>
0
\qquad
\text{for all }g\in[\epsilon_g,1],
\]
the coverage gradient remains nonzero even when the anchor constraint is satisfied and $\mu_t=0$. Moreover, the gradient is larger when coverage is low. In contrast, a multiplicative reward $g\cdot\Rphon$ has coverage derivative $\Rphon$, which can be weak when phonetic compatibility is poor. The log-space formulation therefore preserves anchor-derived guidance independently of the current phonetic score.
\end{remark}

\subsection{Convergence}

\begin{lemma}[KL stabilization]
\label{lem:kl_stab}
For any $\pi_1,\pi_2\in\Pi_\delta$ and any bounded function $f:\mathcal{Y}\times\mathcal{X}\to[0,1]$,
\[
\left|
\E_{\pi_1}[f]
-
\E_{\pi_2}[f]
\right|
\leq
\sqrt{2\delta}.
\]
\end{lemma}

\begin{proof}
By Pinsker's inequality, for any $\pi\in\Pi_\delta$,
\[
\TV(\pi,\piref)
\leq
\sqrt{
\frac{1}{2}
\KL(\pi\|\piref)
}
\leq
\sqrt{\frac{\delta}{2}}.
\]
Since $f\in[0,1]$, the expectation difference is bounded by total variation. Applying the triangle inequality through $\piref$ gives
\begin{equation}
    \begin{aligned}
\left|
\E_{\pi_1}[f]
-
\E_{\pi_2}[f]
\right|
\leq \\
\left|
\E_{\pi_1}[f]
-
\E_{\piref}[f]
\right| \\
+
\left|
\E_{\piref}[f]
-
\E_{\pi_2}[f]
\right|
\leq \\
2\sqrt{\frac{\delta}{2}}
=
\sqrt{2\delta}.
    \end{aligned}
\end{equation}

\end{proof}

\begin{theorem}[Anchor-constraint convergence]
\label{thm:convergence_formal}
Under Assumptions~\ref{assum:bounded}--\ref{assum:variance}, let AP-GRPO run for $T$ dual steps with base step size $\beta>0$ and momentum boost $c\geq0$. Let $\mu^*$ be an optimal dual multiplier.

\textbf{Duality gap.}
\begin{equation}
\begin{aligned}
    \E
    \left[
    \frac{1}{T}
    \sum_{t=1}^{T}
    d(\mu_t)
    \right]
    -
    d(\mu^*)
    \leq \\
    \frac{
    |\mu_0-\mu^*|^2
    }{
    2\beta T
    }
    +
    \frac{
    \beta(1+c)^2
    }{2}
    \left(
    1+\frac{\sigma^2}{G}
    \right)
    +
    \epsilon_p.
    \label{eq:app_duality_gap}
\end{aligned}
\end{equation}

\begin{table*}[t]
\centering
\caption{Qualitative examples from each disease condition. \underline{Underlined} words are extracted audible anchors ($c > \tau$). \textbf{Bold} words in AP-GRPO output indicate correctly recovered inter-anchor content. \sout{Strikethrough} marks errors.}
\label{tab:qualitative}
\small
\begin{tabular}{@{}p{3.8cm}p{8.2cm}cc@{}}
\toprule
Example Audio & Text & $g_{\text{anchor}}$ & $R_{\text{inter}}$ \\
\midrule
\multirow{3}{*}{\shortstack[l]{TORGO F01\\sent\_0003\\(ALS, severe)}}
 & Ref: THE BUBBLE SHEET TEAR FOR DARK KNEW FOR BORN & --- & --- \\
 & Whisper: the bah \underline{sheet} teh fuh dah new fuh \underline{born} & --- & -3.860 \\
 & AP-GRPO: the \textbf{bubble} \underline{sheet} \textbf{tear for dark} \sout{new} \textbf{for} \underline{born} & 1.0 & -1.31 \\
\midrule
\multirow{3}{*}{\shortstack[l]{UASpeech M09\\seg5\\(Cerebral Palsy)}}
 & Ref: we do judith but multiflora in upward many & --- & --- \\
 & Whisper: \underline{we} doo juh buh muh-fla ih \underline{upward} \underline{many} & --- & -2.353 \\
 & AP-GRPO: \underline{we} \textbf{do judith but} multiform \textbf{in} \underline{upward} \underline{many} & 1.0 & -1.722 \\
\midrule
\multirow{3}{*}{\shortstack[l]{ADReSSo\\depaul1b\_137\\(Dementia)}}
 & Ref: like my sister who's almost eight years older than me when she was just born she probably didn't be able to say all sorts of things by the time she was two or whatever & --- & --- \\
 & Whisper: like..lai...like \underline{my sister} who's \underline{almost eight years older} than me wu wu en she she was \underline{just born she probably} didn't be able to \underline{say all sorts of things by the time} she was \underline{two} or or what & --- & -2.793 \\
 & AP-GRPO: like \underline{my sister} who's \underline{almost eight years older} than me \textbf{when she} was \underline{just born she probably} didn't be able to \underline{say all sorts of things by the time} she was \underline{two} \textbf{or whatever} & 1.0 & -1.956 \\
\midrule
\multirow{3}{*}{\shortstack[l]{SJTU-PD DL51\\(Parkinson's)}}
 & Ref: It's very rewarding working with clients and helping them establish their values and their vision. & --- & --- \\
 & Whisper: \underline{it's very rewarding} work in wait \underline{clients} and helping \sout{the} \underline{establish their values and their vision} & --- & -1.635 \\
 & AP-GRPO: \underline{it's very rewarding} \textbf{working with} \underline{clients} \textbf{and helping them} \underline{establish their values and their vision} & 1.0 & -1.356 \\
\bottomrule
\end{tabular}
\end{table*}

\textbf{Constraint violation.}
If $\mu^*>0$, the time-averaged anchor-coverage violation satisfies
\begin{equation}
    \alpha-\E_{\bar\pi}[\ganchor]
    \leq
    O
    \left(
    \frac{
    (1+c)\sqrt{1+\sigma^2/G}
    }{
    \mu^*\sqrt{T}
    }
    \right)
    +
    \frac{\epsilon_p}{\mu^*},
    \label{eq:app_constraint_violation}
\end{equation}
where $\bar\pi$ denotes the standard time-averaged primal iterate.
\end{theorem}

\begin{proof}
Let
\[
h_t
:=
\alpha-\hat g_t
\]
denote the stochastic dual descent direction, where $\hat g_t$ is the batch estimate of $\E_{\pi_{\theta_t}}[\ganchor]$. The exact counterpart is
\[
\bar h_t
:=
\alpha-\E_{\pi^*_{\mu_t}}[\ganchor].
\]
Equivalently, the subgradient of the convex dual function is
\[
s_t
:=
\E_{\pi^*_{\mu_t}}[\ganchor]-\alpha
=
-\bar h_t.
\]

\textbf{Step 1: One-step progress.}
By non-expansiveness of projection onto $\mathbb{R}_+$,
\begin{equation}
    |\mu_{t+1}-\mu^*|^2
    \leq
    |\mu_t-\mu^*|^2
    +
    2\beta_t^{\mathrm{eff}}h_t(\mu_t-\mu^*)
    +
    (\beta_t^{\mathrm{eff}})^2h_t^2.
    \label{eq:app_one_step}
\end{equation}

\textbf{Step 2: Subgradient inequality.}
Since $d$ is convex,
\[
d(\mu^*)
\geq
d(\mu_t)
+
s_t(\mu^*-\mu_t).
\]
Rearranging and using $s_t=-\bar h_t$ gives
\[
\bar h_t(\mu_t-\mu^*)
\leq
-
\left(
d(\mu_t)-d(\mu^*)
\right).
\]

\textbf{Step 3: Stochastic decomposition.}
The stochastic direction decomposes as
\[
h_t
=
\bar h_t
+
\epsilon_t^{\mathrm{primal}}
+
\xi_t,
\]
where $\epsilon_t^{\mathrm{primal}}$ captures primal inexactness and $\xi_t$ is zero-mean sampling noise. Under Assumptions~\ref{assum:iid}--\ref{assum:variance},
\[
\E[\xi_t\mid\theta_t]=0,
\qquad
\E[\xi_t^2]\leq \frac{\sigma^2}{G}.
\]
Since $\ganchor\in[0,1]$, the descent direction is bounded, and
\[
\E[h_t^2]
\leq
1+\frac{\sigma^2}{G}
\]
up to the additive primal approximation term already absorbed into $\epsilon_p$.

\textbf{Step 4: Telescoping.}
The effective step size satisfies
\[
\beta_t^{\mathrm{eff}}
\in
[\beta,\beta(1+c)].
\]
Taking expectations in Eq.~\eqref{eq:app_one_step}, applying the subgradient inequality, using
\[
\E[\xi_t(\mu_t-\mu^*)]=0,
\]
and summing over $t=1,\ldots,T$ yields
\begin{equation}
    \begin{aligned}
\E
\left[
\sum_{t=1}^{T}
d(\mu_t)-d(\mu^*)
\right]
\leq
\frac{
|\mu_0-\mu^*|^2
}{
2\beta
}
+ \\
\frac{
T\beta(1+c)^2
}{2}
\left(
1+\frac{\sigma^2}{G}
\right)
+
T\epsilon_p.
    \end{aligned}
\end{equation}

Dividing by $T$ gives Eq.~\eqref{eq:app_duality_gap}.

For the constraint-violation bound, weak duality and complementary slackness for an active constraint imply that the anchor violation is controlled by the duality gap scaled by $\mu^*$. Optimizing the bound with respect to $\beta$ yields the stated $O(1/\sqrt{T})$ rate in Eq.~\eqref{eq:app_constraint_violation}.
\end{proof}

\begin{remark}[Scheduled KL strengthens damping]
\label{rem:scheduled_kl}
With
\[
\eta_t
=
\eta_0
+
(\eta_T-\eta_0)t/T
\qquad
\text{and}
\qquad
\eta_0<\eta_T,
\]
the effective KL budget $\delta(t)$ is non-increasing over training. Lemma~\ref{lem:kl_stab} therefore holds at each step with $\delta$ replaced by $\delta(t)\leq\delta(0)$, so the reward-expectation variation bound tightens as training progresses. Consequently, the dual variable cannot oscillate faster than one reversal per
\[
\left\lceil
\frac{1}{\sqrt{2\delta(t)}}
\right\rceil
\]
steps, and this lower bound on the oscillation half-period increases as $\delta(t)$ decreases.
\end{remark}



\section{Qualitative Examples}
\label{app:qualitative}

Figure~\ref{fig:overview} shows three qualitative examples at the training. 
We will provide per-disease qualitative examples showing the Whisper output, reference transcript, timestamped anchors, and reconstructions from each method. Following data access consent, we will not release the audio as we are not the owners. Instead, we will provide index for readers to find them on the original datasets.


\end{document}